\documentclass[a4paper,11pt]{article}
\usepackage{url}

\usepackage{amsmath, amssymb}
\usepackage{amsthm}
\usepackage[OT1]{fontenc}

\usepackage{sectsty,textcase}

\usepackage[margin=1.1in]{geometry}
\usepackage{graphicx,amssymb,amstext,amsmath}
\usepackage{amsmath,amsfonts,amssymb,amsthm}
\usepackage{mathtools}
\usepackage{float}
\usepackage[update]{epstopdf}
\usepackage{amsthm}
\usepackage{xcolor}
\usepackage{fancybox}
\usepackage{framed}

\newtheorem{theorem}{Theorem}[section]
\newtheorem{definition}[theorem]{Definition}
\newtheorem{lemma}[theorem]{Lemma}

\newtheorem{remark}[theorem]{Remark}

\newtheorem{corollary}[theorem]{Corollary}

\newtheorem{condition}[theorem]{Condition}

\newcommand{\rr}{\rightarrow}
\newcommand{\comment}[1]{}
\newcommand{\td}[1]{\tilde{#1}}

\newcommand{\Prob}{\mathbb{P}}
\newcommand{\E}{\mathbb{E}}
\newcommand{\R}{\mathbb{R}}

\newcommand{\F}{\mathbb{F}}
\newcommand{\Ft}{\mathcal{F}}

\newcommand{\A}{\mathcal{A}}

\newcommand{\Lp}{\mathcal{L}}

\newcommand{\Var}{\text{var}}

\graphicspath{{./}{Fig/}} 

\usepackage{xcolor}

\title{Stationarity of Bivariate Dynamic Contagion Processes}
\author{Angelos Dassios\footnote{Department of Statistics, London School of Economics and Political Science, London WC2A 2AE, UK. Email: a.dassios@lse.ac.uk}, ~ Xin Dong\footnote{Department of Mathematics, Imperial College, London SW7 2AZ, UK. Email: x.dong10@imperial.ac.uk}}
\date{}

\begin{document}

\maketitle

\begin{abstract}
The Bivariate Dynamic Contagion Processes (BDCP) are a broad class of bivariate
point processes    characterized  by the intensities  as a general class
of piecewise deterministic Markov processes. The BDCP describes a rich dynamic
structure where  the system is under the influence of both external and internal
 factors modelled by a shot-noise Cox process and a generalized Hawkes process
respectively. In this paper we  mainly address the stationarity issue for
the BDCP, which is important in  applications. We investigate the stationary
distribution by applying the the Markov theory on the branching system approximation
representation of the BDCP. We
find the  condition under which there exists a unique stationary distribution
of the BDCP intensity and the resulting BDCP has stationary increments. Moments
of the stationary intensity are provided by using the Markov property.
\end{abstract}






\section{Introduction}

Multivariate point processes are used to model event arrivals of different
types within a system. There are many potential applications; stochastic
models are needed for events such as company bankruptcies, insurance claim
arrivals, disease incidence, machine failures and many others. Modelling
the point process in way that capture a rich dependence structure becomes
an essential problem. Moreover, stationarity is an important and common assumption
in many statistical applications  and is also one
of the most important probability properties in stochastic process
study. Hence, the existence of a version of a point process with stationary
increments and stationary intensity needs to be investigated.

In order to describe a system with a rich dynamic reflecting both external
impact and internal contagion effect, we introduce the Bivariate Dynamic
Contagion Process (BDCP). The BDCP is a broad family of bivariate point processes
with intensity processes specified as non-diffusion Piecewise Deterministic
Markov Processes (PDMP) studied by Davis\cite{Davis1984} but also incorporating
a feedback mechanism (internal contagion). The BDCP covers two distinct important
classes of point processes. The first class consists of shot-noise Cox processes
that usually describe point process systems under the impact from external
factors. Shot-noise Cox processes are studied by Cox and
Isham\cite{Cox1980},  M{\o}ller\cite{Moller2003}, Dassios and Jang\cite{DJ2003},
and Kl{\"u}ppelberg and Mikosch\cite{Kluppelberg1995} for example. The class
has a wide range of applications. For instance, it is adopted in modelling
insurance claim arrivals and ruin probabilities by Altmann et al.\cite{Altmann2008},
Albrecher
and Asmussen\cite{Albrecher2006} and Macci and Torrisi\cite{Torrisi2011}.
 The second class consists of generalized Hawkes process which have mutually-exciting
intensities. In this class, jumps in the point process bring the internal
feedback into the underlying intensity process and the impact factor is modelled
by upward jumps with random marks. This class is capable to model clustering
and contagion effects. Hawkes processes are introduced by Hawkes and Oakes\cite{HO1974},
and they
are studied
by Daley and Vere-Jones\cite{DaleyVJones2007}, Liniger\cite{Liniger2009}
and Embrechts et al.\cite{Embrechts2011}. Recently Hawkes processes are extensively
applied
in finance and insurance modelling as Hautsch\cite{BH2009}, A{\"\i}t-Sahalia
et al.\cite{Ait2010}, Bacry et al.\cite{BDHM2013} and Errais et
al.\cite{EK2010}. 

In the univariate case,  Dassios
and Zhao\cite{DZ2011} introduce Univariate Dynamic Contagion Processes (UDCP)
that include impact by both external and internal factors. They can be used
in credit risk and insurance modelling as in Dassios and Zhao\cite{DassiosZhao2011Credit}
and \cite{DZ2012}.  However, in practice a univariate model is not sufficient
to model a heterogeneous population with rich dependence structure. In order
to address this issue, we introduce a bivariate system. The dependency between
marginals can be characterized in a few ways. Dassios and Jang\cite{DassiosJang2013}
 studied a bivariate system with correlated shot-noise components and self-exciting
components but the dependence due to the cross-exciting contagion effect
is missing.
We address this in our definition of BDCP. Note that the cross-exciting dependency
introduces a loop structure that makes the system difficult to be decoupled.
Hence it is fundamentally different from the univariate case. Moreover, the
BDCP is also different from the bivariate Hawkes process since it is not
obvious how the additional external factor modelled by a Cox process and
randomness from jump sizes affect the probability properties of the system.
In principle, applications of Hawkes or shot-noise processes can be extended
using BDCP to incorporate richer structure.

Once the dynamic of the point process is specified, stationarity becomes
an important issue to address. It is a reasonable assumption and many problems
can be
simplified based on it. With stationarity of the
intensity, Dassios
and Zhao\cite{DZ2012} discuss the ruin probability in insurance modelling
using the UDCP. Moreover, Dassios and Dong\cite{DD2013}    explore the diffusion
approximation
of BDCP with filtering applications based on the stationarity assumption.

Previously, Costa\cite{Costa1990} discusses the stationarity condition of
piecewise deterministic Markov processes in general. Dassios
and Zhao\cite{DZ2011} show the existence of a stationary distribution for
UDCP. Br\'emaud and Massoulie\cite{BM1996} discuss the stationarity and stability
of Hawkes processes. Furthermore, in our
approach the BDCP can be seen as a limit of finite dimensional affine processes
where the dimension is tending to infinity. This is itself an interesting
case that is not dealt in the affine literature so far. One can look into
Duffie et al.\cite{Duffie2002}, Keller-Ressel et
al.\cite{Keller2011} and a few others for the research on affine processes.
We note that the stationarity results are only available in a few cases of
diffusion affine processes. For example, the discussion of the stationarity
of two-factor diffusive affine processes can be found in Glasserman and Kim\cite{Glasserman2010}
and Barczy et al.\cite{Barczy2013}.

In this paper, the analysis of the BDCP intensity is based on the approximation
of the  finite branching system resulting from the cluster-based representation.
We apply the PDMP theory developed by Davis\cite{Davis1984} on
the branching system to explore the limiting distribution as $t\rr\infty$.
Moreover, the
link between the stationary distribution and limiting distribution is explored.

The definition and the cluster representation of the BDCP are provided in
Section~\ref{sec: ModelRepresentation}, where we introduce a finite system
 $(\lambda^{1,n},\lambda^{2,n})$ that approximates the BDCP intensity $(\lambda^1,\lambda^2)$
and  a finite joint system  $(\Lambda^{(1)},\ldots,\Lambda^{( 2n)})$ resulting
from a dimension translation. Then in Section~\ref{sec: FiniteSystem}, starting
from the finite joint system $(\Lambda^{(1)},\ldots,\Lambda^{(
2n)})$ which is a Markov process and  de-coupled, we apply the PDMP theory
  to obtain the limiting distribution as $t\rr\infty$ in terms of the Laplace
transform. With the the branching  system approximation as $n\rr \infty$,
the condition of the existence of the limiting distribution of  $(\lambda^1,\lambda^2)$
is investigated. The limiting distribution result can be found in Theorem~\ref{thm:
LimitingDistributionAll}   and the existence condition is the Condition~\ref{asm:
stationaryCondition}. In Section~\ref{sec: StationaryDistribution}, again
starting from $(\Lambda^{(1)},\ldots,\Lambda^{(
2n)})$, we provide a stationarity condition in Lemma~\ref{thm: statCondTruncSys},
which is in terms of the Laplace transform based on the Markov theory. As
we have found in Section~\ref{sec:
FiniteSystem} the limiting distribution of the finite joint system that is
a natural candidate, we confirm the limiting distribution is also the stationary
distribution for $(\Lambda^{(1)},\ldots,\Lambda^{(
2n)})$ and $(\lambda^{1,n},\lambda^{2,n})$ in Theorem~\ref{prop: AsyEqStat}
and Corollary~\ref{cor: Stationarity_BDCP_N} respectively. The approximation
argument is applied to conclude the stationarity of tbe BDCP intensity $(\lambda^1,\lambda^2)$
in Theorem~\ref{thm: stationarity_BDCP} and also BDCB $(N^1,N^2)$ in Corollary~\ref{cor:
Stationarity_BDCP_N}. In Section~\ref{sec: StationaryMoments},  we  provide
the stationary moments of the intensity process  $(\lambda^1,\lambda^2)$.
We conclude in Section~\ref{sec: Conclusion}.

\section{The Model}\label{sec: ModelRepresentation}

\subsection{Model}\label{sec: Model}
Let $(\Omega,\Ft,\Prob)$ be a probability space, on which we introduce the
Bivariate Dynamic Contagion Processes (BDCP) as a class
of bivariate point processes $N_{t} = (N^1_t,N^2_t)$ defined on $\R_{+}$.
Let $\F$ be a filtration such that $N$ is $\F$-adapted. For $i=1,2$, $$N^i_t
= \sum_{n\geq 1}1_{\{T^i_n\leq t\}}$$ where $\{T^i_n\}_{n\geq0}$ are orderly
$\F$-stopping times representing event arrival times with $T_0=0$. By the
Doob-Meyer decomposition, there exists a unique non-decreasing process $A$
starting at $0$, such that $N - A$ is an $\F$-local martingale. Suppose that
there exists a non-negative, $\F$-predictable and integrable intensity process
$\lambda$, s.t. for every $t\geq 0$, $A_t = \int_0^t \lambda_s ds$ $a.s.$
 We assume that the filtration $\F$ satisfies the usual condition. 

For $N_{t}$ as a BDCP, its intensity processes
$\lambda_t = (\lambda^1_t,\lambda^2_t)$ is specified as a Piecewise Deterministic
Markov Processes (PDMP) introduced by Davis\cite{Davis1984}.

\begin{definition}[Bivariate Dynamic Contagion Processes (Intensity-based)]
The intensity $ \lambda_t = (\lambda^1_t,\lambda^2_t)$ of the BDCP $N_t =(N^1_t,N^2_t)$
with $t\in
\mathbb{R}_+$ is defined by  
\begin{equation}\label{eqn: Intensity_BDCP_Def}
\begin{aligned}
\lambda^1_t &= \lambda_0^1 e^{-\delta_1 t} + \sum_{S^{1}_j < t}Y^1_{j} e^{-\delta_1
(t-S^1_j)} + \sum_{T^1_j<t}^{}Z^{1,1}_{j}e^{-\delta_1 (t-T^1_j)}+\sum_{T^2_j<t}^{}Z^{1,2}_{j}e^{-\delta_1
(t-T^2_j)},\\
\lambda^2_t &= \lambda_0^2 e^{-\delta_2 t} + \sum_{S^2_j<t}^{}Y^2_{j} e^{-\delta_2
(t-S^2_j)} + \sum_{T^1_j<t}^{}Z^{2,1}_{j}e^{-\delta_2 (t-T^1_i)} + \sum_{T^2_j<t}^{}Z^{2,2}_{j}e^{-\delta_2
(t-T^2_j)}.
\end{aligned}
\end{equation}
For $k,k'=1,2$, 
\begin{itemize}
\item $\lambda^k_0 \geq 0$ is the initial intensity at time $t=0$;
\item $\delta_k >0$ is the constant rate of exponential decay;
\item $\{S^k_j\}_{j=1,2,\ldots}$ are jump times of $M^k_t$, which is  a Poisson
process with the constant intensity $\rho_k$.
 $\{Y^k_j\}_{j=1,2,\ldots}$
are associated i.i.d jump sizes with the distribution function $H_k(\cdot)$
and Laplace transform $\hat{h}_k(\cdot)$;
\item $\{T^k_j\}_j$ are jump times of $N^k$, and $\{Z^{k,k'}_j\}_{j=1,2,\ldots}$
are i.i.d. jump sizes with  distribution function $G_{k,k'}(\cdot)$ and
Laplace transform $\hat{g}_{k,k'}(\cdot)$;

\item $\{S^k_j\}_{j=1,2,\ldots}$ and $\{T^k_j\}_{j=1,2,\ldots}$ are independent
of $\{Y^k_j\}_{j=1,2,\ldots}$ and $\{Z^{k,k'}_j\}_{j =1,2,\ldots}$.  
\end{itemize}
\label{def:BDCP_intensity}
\end{definition}

Due to the exponential decay, $\lambda_t = (\lambda^1_t,\lambda^2_t)$
is a piecewise deterministic Markov process (PDMP). 
For $k, k'=1,2$, the marked point process  $\sum_{S^{k}_j < t}Y^k_{j} e^{-\delta_k(t-S^k_j)}$
characterizes the dependence on an external factor. $\sum_{T^{k}_j < t}Z^{k,k}_j
e^{-\delta_k
(t-T^k_i)}$ and $\sum_{T^{k'}_j < t}Z^{k,k'}_{j}e^{-\delta_k (t-T^{k'}_j)}$
for $k'\neq k$ characterize the internal dependence due to the self-exciting
and cross-exciting effect respectively. Note that the impact factor  is modelled
by random marks that are independent of $N$.

From the intensity-based definition, the BDCP is a broad class of point processes
covering two distinct and important point process classes. The first class
is the   shot-noise Cox
processes that can be obtained by setting $Z^{k,k'}_j \equiv 0$ for all $j\geq
1$, $k,k'\in\{1,2\}
$. The second class is the bivariate Hawkes processes with exponential decay
 obtained by setting $Y^k_j \equiv 0$ and $Z^{k,k'}_j$ as constants for all
$j$, $k,k'\in\{1,2\}$. We assume the following condition always holds.

\begin{condition}\label{cond: finite_mean_marks}
\begin{enumerate}
\item[(C1)] For $k,k'=1,2$, all random marks $\{Y^k_j\}_{j}$ and $\{Z^{k,k'}_j\}_{j}$
 have the finite first moments. i.e. $\mu_{H^k}$, $\mu_{G^{k,k'}}$ are finite.\\
\end{enumerate}
\end{condition}
One can easily check that under $(C1)$ in Condition~\ref{cond: finite_mean_marks},
$\int_0^t \lambda_s ds<\infty$ $a.s.$ for every $t\geq 0$, thus the BDCP
$N$ is non-explosive. \\

\subsection{Stationarity}
First recall the definition of the stationary distribution and
the stationary process for Markov processes       (e.g. Section 9.4 in Ethier
and Kurtz\cite{EK1986}). Suppose that a martingale problem for  $\A$ on
state space $E$ is well defined, then $\mu$ as a probability measure on $E$
is a stationary distribution of $\A$ if every solution $X$ of the martingale
problem for $(\A,\mu)$ is a stationary process, that is, if $\Prob\left(X_{t+s_1}\in
\Gamma_1,\ldots, X_{t+s_k}\in\Gamma_k \right)$ is independent of $t\geq 0$
for all $k\geq 1$, $0\leq s_1 \leq \cdots \leq  s_k$, and $\Gamma_1,\ldots,
\Gamma_k \in \mathcal{B}(E)$. Moreover, $\mu$ is a stationary distribution
for $\A$ if and only if $X_t$ has the distribution $\mu$ for all $t\geq 0$.

For a Markov process $X$ that solves the martingale problem for $(\mathcal{A},\mu)$
with the domain $\mathcal{D}(\mathcal{A})$, Proposition~9.2, Chapter 4 in
Ethier and Kurtz \cite{EK1986} provides
a stationarity theorem:
\begin{equation} \label{eqn: Stationarity_Markov}
\text{The stationary distribution $\mu$ exists if and only if for $f\in\mathcal{D}(\A)$,
} \int_E \A f(x) d\mu (x) = 0.
\end{equation}

In this paper, we find a sufficient condition $(C2)$ in Condition~\ref{asm:
stationaryCondition}, under which, there exists a unique stationary distribution
of $\lambda$ and also a unique stationary distribution for the BDCP $N$ in
Section~\ref{sec: StationaryDistribution}.

\subsection{The Branching Structure}\label{sec: Representation_Branching}


Note that the intensity process $\lambda_t = (\lambda^1_t,\lambda^2_t)$ has
a representation as a cluster process with a branching structure.
The representation is helpful in the following analysis.

\begin{definition}[Bivariate Dynamic Contagion Processes (Cluster-based)]

A bivariate dynamic contagion processes $N=(N^1,N^2)$ is a two-type Poisson
cluster
process $(C^1, C^2)$ with the branching interpretation as follows:
\begin{itemize}

\item For $k=1,2$, the cluster centers of type $k$ are immigrants arrived
at $\{T^{k,(0)}_m\}_{m=1,2,\ldots}$ as a shot-noise process with intensity
$\lambda^{k,(0)}_t =  \lambda^{k}_0
e^{-\delta_k t} +  \sum_{m=1}^{M^{k}_t} Y^k_m e^{-\delta_k (t-S^{k}_m)}$,
where $\{S^{k}_m\}_m := T^{k,(0)}_m $. 
\item Each cluster center of type $k$ at $T^{k,(0)}_m$ generates a cluster
$C^k_m$ consisting of events of type $k$ and also the cluster center itself.
Then the cluster $C^k = \cup_{m=1}^\infty C^k_m$. In branching term, each
$C^k_m$ is the set of type $k$ immigrant arrived at $T^{k,(0)}_m$ and its
offspring of all generations. 

\item Denote the set of the n-th generation offspring of type $k$ in the
cluster $C^k_m$ as $C^{k,(n)}_m$, then the collection of the n-th generation
offspring of type $k$ from all clusters is $C^{k,(n)} = \cup_{m=1}^\infty
C^{k,(n)}_m$. Denote the offspring birth process in $C^{k,(n)}$ as $N^{k,(n)}_t$
with the arrival times $\left\{T^{k,(n)}_j \right\}_j$ and the intensity
$\lambda^{k,(n)}_t$. Recursively, the $(n+1)$-th generation from all clusters
are generated with the intensity $$\lambda^{k,(n+1)}_t = \sum_{j=1}^{N^{1,(n)}_t}
Z_j^{k,1,(n)}
e^{-\delta_k (t-T^{1,(n)}_j)} +  \sum_{j=1}^{N^{2,(n)}_t} Z_j^{k,2,(n)}
e^{-\delta_k (t-T^{2,(n)}_j)},$$ where for $k,k' = 1,2$, the random marks
$Z_j^{k,k',(n)}$ are independent copies of $Z_j^{k,k'}$ for all $n$.

\item Collect all individuals of type $k$ up to the n-th generation from
all clusters, denoted as $C^{k,n}$, then $C^{k,n} = \cup_{j=1}^n C^{k,(j)}$.
The offspring birth process in $C^{k,n}$ is $N^{k,n}_t$ with birth times
$\left\{T^{k,n}_j\right\}_j$ and the intensity process $\lambda^{k,n}_t$.
Hence, we have
$$N^{k,n}_t = \sum_{i=0}^n N^{k,(i)}_t,\quad \lambda^{k,n}_t = \sum_{i=0}^k
\lambda^{k,(i)}_t,\quad \left\{T^{k,n}_j\right\}_j = \cup_{i=0}^n \left\{T^{k,(i)}_j\right\}_j.$$
\end{itemize}

\end{definition}

Clearly,
\begin{equation*}
C^k = \lim_{n\rr \infty}C^{k,n}=\lim_{n\rr \infty} \cup_{j=1}^n C^{k,(j)}
= \lim_{n\rr \infty} \cup_{j=1}^n
\cup_{m=1}^{\infty} C^{k,(j)}_m. 
\end{equation*} 

By construction, all clusters $\{C^k_m\}_{m=1,2,\ldots}$ are independent.
Moreover, we have pathwise,
\begin{eqnarray*}
N^{k}_t &=& \lim_{n\rr\infty}N^{k,n}_t = \lim_{n\rr\infty} \sum_{i=0}^n N^{k,(i)}_t,
\\
\lambda^{k}_t &=& \lim_{n\rr\infty}\lambda^{k,n}_t = \lim_{n\rr\infty} \sum_{i=0}^n
\lambda^{k,(i)}_t.
\end{eqnarray*} 

We call $N^n = (N^{1,n},N^{2,n})$ as the BDCP with truncated finite system
 that is with intensity $\lambda^n = (\lambda^{1,n},\lambda^{2,n})$.\\

\comment{
\begin{eqnarray*}
\lim_{n\rr\infty}\lambda^{j,n}_t(\omega) = \lambda^j_t(\omega) \quad \text{for
all $t\geq 0$.}
\end{eqnarray*}
}

\textbf{From the bivariate system to the univariate system}
In the next section, we will apply the Markov theory to analyse the joint
distribution of  generations of type 1 and 2. To simplify the multi-type
problem, we merge the bivariate branching
system into a univariate system, such that the $i$-th generation of type
1 and type 2 offspring become the $(2i-1)$-th and $2i$-th generation in the
new univariate system. Denote the birth time of the $n$-th generation in
the new system as $\{T^{(n)}_j\}_j$ and the counting process as $N^{(n)}_t$
with intensity $\Lambda^{(n)}_t$, then for $i=1,2,\ldots,$ 
\begin{equation*}
\Lambda^{(2i-1)}_t = \lambda^{1,(i)}_t,\quad \Lambda^{(2i)}_t = \lambda^{2,(i)}_t.
\end{equation*}

Hence, 
\begin{eqnarray*}
\Lambda^{(1)}_t &=& \lambda^1_0  e^{-\delta_1 t} + \sum_{i=1}^{M^{1}_t}Y^1_i
e^{-\delta_1 \left(t-T^{1,(0)}_i\right)},\\
\Lambda^{(2)}_t &=& \lambda^2_0  e^{-\delta_2 t} + \sum_{i=1}^{M^{2}_t}Y^2_i
e^{-\delta_2 \left(t-T^{2,(0)}_i\right)},\\
\Lambda^{(2i+1)}_t &=& \sum_{j=1}^{N^{(2i-1)}_t}Z^{1,1}_j e^{-\delta_1 \left(t-T^{(2i-1)}_j\right)}
+ \sum_{j=1}^{N^{(2i)}_t}Z^{1,2}_j e^{-\delta_1 \left(t-T^{(2i)}_j\right)},\\
\Lambda^{(2i+2)}_t &=& \sum_{j=1}^{N^{(2i-1)}_t}Z^{2,1}_j e^{-\delta_2 \left(t-T^{(2i-1)}_j\right)}
+ \sum_{j=1}^{N^{(2i)}_t}Z^{2,2}_j e^{-\delta_2 \left(t-T^{(2i)}_j\right)}.
\end{eqnarray*}

Moreover, by construction, the original branching system is recovered by:
\begin{eqnarray*}
\lambda^{1,n}_t &=& \sum_{i=1}^{n} \Lambda^{(2i-1)}_t, \quad 
N^{1,n}_t = \sum_{i=1}^{n} N^{(2i-1)}_t, \\
\lambda^{2,n}_t &=& \sum_{i=1}^{n} \Lambda^{(2i)}_t,\quad 
N^{2,n}_t = \sum_{i=1}^{n} N^{(2i)}_t.
\end{eqnarray*}

Hence the original bivariate system with truncation up to $n$-th generation
is transformed into a univariate system with truncation up to $m$-th generation
with $m=2n$. We call $m$ and $n$ as the system index for the transformed
and original system respectively.

Denote the limiting distribution when $t\rr \infty$ and the stationary distribution:
\begin{itemize}
\item $\left(\Lambda^{(1)},\ldots, \Lambda^{(m)} \right)$: $\pi^m_A$ and
$\pi^m_S$
\item $\left(\lambda^{1,n}, \lambda^{2,n} \right)$: $\mu^n_A$ and $\mu^n_S$
\item $\left(\lambda^1,\lambda^2\right)$: $\mu^*_A$ and $\mu^*_S$
\end{itemize}

\section{Markov Property and Limiting Distributions } \label{sec: FiniteSystem}
We use the Markov property and the PDMP theory to explore the limiting distributions
of the intensities in this section. In the next section, we build the relationship
between the limiting distribution and the stationary distribution.

\subsection{Markov Property}\label{sec: MarkovProperty_FiniteSystem}

Though the intensity $(\lambda^1_t,\lambda^2_t)$ is a Markov process, it
is difficult to explore the stationarity using the PDMP theory as they are
coupled due to the cross-exciting components. Moreover, it is not obvious
how we can find the condition on the intensity process such that the existence
and stationarity can be analyzed. 
Hence, the finite system with the branching
structure introduced above will be used for the stationarity analysis.

The finite system $(\lambda^{1,n}_t,\lambda^{2,n}_t)$ and is not Markovian,
but the joint system $\left(\Lambda^{(1)}_t,\Lambda^{(2)}_t,\ldots,\Lambda^{(m)}_t\right)$
is. The generator of $\left(t,\Lambda^{(1)}_t,\Lambda^{(2)}_t,\ldots,\Lambda^{(m)}_t\right)$
is $ \A_m$ with  domain $\mathcal{D}\left( \A_m
\right)$. For any $f\in \mathcal{D}\left( \A_m
\right)$,
\begin{eqnarray*}
&&\A_m f(t, \lambda_1,\lambda_2,\ldots,\lambda_m) \nonumber\\
&=& \frac{\partial f}{\partial t} -\sum_{i=1}^m \delta_i \lambda_i \frac{\partial
f}{\partial \lambda_i}
 + \rho_1 \left[\int_0^\infty f(\lambda + e_1 y)dH_1(y) -f(\lambda) \right]
+ \rho_2 \left[\int_0^\infty f(\lambda + e_2 y)dH_2(y) -f(\lambda) \right]
\nonumber\\
&& + \sum_{k=1}^{n-1} \lambda_{2k-1}\left[ \left(\int_0^\infty f(\lambda+
e_{2k+1}z)dG_{1,1}(z) -f(\lambda) \right) + \left(\int_0^\infty f(\lambda
+ e_{2k+2}z)dG_{2,1}(z) -f(\lambda) \right) \right]\nonumber\\
&& + \sum_{k=1}^{n-1} \lambda_{2k}\left[ \left(\int_0^\infty f(\lambda+ e_{2k+1}
z)dG_{1,2}(z) -f(\lambda) \right) + \left(\int_0^\infty f(\lambda + e_{2k+2}z)dG_{2,2}(z)
-f(\lambda) \right) \right],
\end{eqnarray*}
where $\mathbf{\lambda}:= (\lambda_1,\ldots,\lambda_m)$ and $e_i := (0,\ldots,0,1,0,\ldots)
\in \R^m$ where only the $i$-th element is $1$ and others are all $0$.

Take $f\left(t,\Lambda^{(1)}_t,\cdots,\Lambda^{(m)}_t\right) = e^{-B_1(t)\Lambda^{(1)}_t-\cdots-B_m(t)\Lambda^{(m)}_t
+ c_m(t)}$ and suppose it is a martingale. Consider for any $T>0$ and assume
that $B_i(T) = v_i$ and $c_m(0) = 0$, then the Laplace transform
of $(\Lambda^{(1)}_T,\ldots, \Lambda^{(m)}_T)$ conditional on the initial
condition $\Lambda_0
= (\Lambda^{(1)}_0,\ldots, \Lambda^{(m)}_0)$ at $(v_1,\ldots,v_m)\in \R^m_+$
is
\begin{eqnarray*}
 \E_0\left[ e^{-v_1\Lambda^{(1)}_T -
\cdots - v_m\Lambda^{(m)}_T } \right] = \E_0\left[ e^{-B_1(T)\Lambda^{(1)}_T-\cdots
- B_m(T)\Lambda^{(m)}_T } \right] =  e^{-B_1(0)\Lambda^{(1)}_0-\cdots-B_m(0)\Lambda^{(m)}_0
- c_m(T)}. 
\end{eqnarray*}

A sufficient condition for $f$ to be a martingale is $\A_m f(t, \lambda_1,\lambda_2,\ldots,\lambda_m)
= 0$ for any $t$, $\{\lambda_i\}_{i=1}^{m}$, $\rho_1$, $\rho_2$ on $\R_+$,
i.e., 
\begin{eqnarray*}
0=\frac{\A_m f}{f} &=& \sum_{k=1}^{n-1} \lambda_{2k-1}\left[-\dot{B}_{2k-1}(t)+\delta_1
B_{2k-1}(t)+ \left(\hat{g}_{1,1}(B_{2k+1}(t))-1\right)+\left(\hat{g}_{2,1}(B_{2k+2}(t))-1\right)\right]
\\
&&+\sum_{k=1}^{n-1} \lambda_{2k}\left[-\dot{B}_{2k}(t)+\delta_2 B_{2k}(t)+
\left(\hat{g}_{1,2}(B_{2k+1}(t))-1\right)+\left(\hat{g}_{2,2}(B_{2k+2}(t))-1\right)\right]\\
&&+\lambda_{2n-1}\left[-\dot{B}_{2n-1}(t)+\delta_1 B_{2n-1}(t) \right]
+ \lambda_{2n}\left[-\dot{B}_{2n}(t)+\delta_2 B_{2n}(t) \right]\\
&& + \dot{c}_m(t)+\rho_1 \left( \hat{h}_1(B_1(t))-1 \right)+\rho_2 \left(
\hat{h}_2(B_2(t))-1 \right).
\end{eqnarray*}

Therefore, the sequence of functions $\left(B_i(t)\right)_{i=1}^{m}$ solves
the backward recursive ODE system ($k=1,\ldots, n-1$)
\begin{eqnarray*}
-\dot{B}_{2n}(t)+\delta_1 B_{2n}(t) &=& 0,\quad B_{2n}(T) = v_{2n},\\
-\dot{B}_{2n-1}(t)+\delta_2 B_{2n-1}(t) &=& 0,\quad B_{2n-1}(T) = v_{2n-1},\\
-\dot{B}_{2k-1}(t)+\delta_1 B_{2k-1}(t)+ \left(\hat{g}_{1,1}(B_{2k+1}(t))-1\right)+\left(\hat{g}_{2,1}(B_{2k+2}(t))-1\right)
&=&0,\quad B_{2k-1}(T) = v_{2k-1},\\
-\dot{B}_{2k}(t)+\delta_2 B_{2k}(t)+ \left(\hat{g}_{1,2}(B_{2k+1}(t))-1\right)+\left(\hat{g}_{2,2}(B_{2k+2}(t))-1\right)
&=&0,\quad B_{2k}(T) = v_{2k}.
\end{eqnarray*}
Moreover,
\begin{eqnarray*}
 \dot{c}_m(t)+\rho_1 \left( \hat{h}_1(B_1(t))-1 \right)+\rho_2 \left( \hat{h}_2(B_2(t))-1
\right) &=& 0,\quad c_m(0) = 0.
\end{eqnarray*}

We transform the system that is backward in the system index and the time
into a forward system by taking
\begin{eqnarray*}
l_k(t) := B_{m+1-k}(T-t)= B_{2n+1-k}(T-t).
\end{eqnarray*}

By construction
 $\Lambda^{(1)}_0 = \lambda^1_0,
\Lambda^{(2)}_0 = \lambda^2_0$ and $\Lambda^{(j)}_0 \equiv 0$ for $j>2$,
then the Laplace transform becomes
\begin{eqnarray}\label{eqn: Laplace_L_m_T_forward}
\E_0\left[ e^{-v_1\Lambda^{(1)}_T -
\cdots - v_m\Lambda^{(m)}_T } \right] =  e^{-l_m(T)\Lambda^{(1)}_0-\cdots-l_1(T)\Lambda^{(m)}_0
- c_m(T)} =e^{-l_{2n}(T)\lambda^{1}_0-l_{2n-1}(T)\lambda^{2}_0
- c_m(T)}  ,
\end{eqnarray}
where $l_i(t)$ and $c_m(t)$ solves the forward ODE system: with $k=1,\ldots,
n-1$,
\begin{equation}\label{eqn: ForwardODE}
\begin{aligned}
\dot{l}_1(t) + \delta_2 l_1(t) =~& 0,\quad l_1(0) = v_{2n},\\
\dot{l}_2(t) + \delta_1 l_2(t) =~& 0,\quad l_2(0) = v_{2n-1},\\
\dot{l}_{2k+1}(t) + \delta_2 l_{2k+1}(t) - \left(1 - \hat{g}_{1,2}(l_{2k}(t))\right)-\left(1-\hat{g}_{2,2}(l_{2k-1}(t))\right)
=~&0,\quad l_{2k+1}(0) = v_{2(n-k)},\\
\dot{l}_{2k+2}(t) + \delta_1 l_{2k+2}(t) - \left(1 - \hat{g}_{1,1}(l_{2k}(t))\right)
- \left(1 - \hat{g}_{2,1}(l_{2k-1}(t))\right) =~&0,\quad l_{2k+2}(0) = v_{2(n-k)-1},\\
\dot{c}_m(t) - \rho_1 \left(1 - \hat{h}_1(l_{2n}(t))\right) - \rho_2 \left(1
- \hat{h}_2(l_{2n-1}(t))\right)=~&0,\quad c_m(0) = 0.
\end{aligned}
\end{equation}

Note that the ODE system (\ref{eqn: ForwardODE}) has a unique and explicit
solution in a recursive form
\begin{eqnarray}\label{eqn: l_function}
 l_1(t) &=& v_{2n}e^{-\delta_2 t},\nonumber\\
 l_2(t) &=& v_{2n-1}e^{-\delta_1 t},\nonumber\\ 
l_{2k+1}(t) &=& v_{2(n-k)}e^{-\delta_2 t} + e^{-\delta_2 t}\int_0^t e^{\delta_2
s}\left[ 1-\hat{g}_{1,2}(l_{2k}(s))+1-\hat{g}_{2,2}(l_{2k-1}(s))\right] ds
,\nonumber\\
l_{2k+2}(t) &=& v_{2(n-k)-1}e^{-\delta_1 t} + e^{-\delta_1 t}\int_0^t e^{\delta_1
s}\left[ 1-\hat{g}_{1,1}(l_{2k}(s))+1-\hat{g}_{2,1}(l_{2k-1}(s))\right] ds.\nonumber\\
\end{eqnarray}

Moreover,
\begin{eqnarray*}
c_m(T) &=& \rho_1 \int_0^T \left[1-\hat{h}_1(l_{2n}(t))\right]dt + \rho_2
\int_0^T \left[1-\hat{h}_2(l_{2n-1}(t))\right]dt.
\end{eqnarray*}

\subsection{Limiting Distributions}\label{sec: LimitingDistribution}

We first introduce the following crucial condition.
\begin{condition}\label{asm:
stationaryCondition}\quad\\ 
\begin{enumerate}
\item[(C2)] The spectral radius of the matrix 
$\begin{bmatrix}
\frac{\mu_{G^{2,2}}}{\delta_2} & \frac{\mu_{G^{1,2}}}{\delta_2}\\ 
\frac{\mu_{G^{2,1}}}{\delta_1} &
\frac{\mu_{G^{1,1}}}{\delta_1}
\end{bmatrix}$ is less than $1$.
\end{enumerate}
\end{condition}

\begin{remark}
The spectral radius condition in the condition above is equivalent to the
following
\begin{eqnarray}\label{eqn: statCondition}
\frac{1}{2}\left(\frac{\mu_{G^{1,1}}}{\delta_1}+\frac{\mu_{G^{2,2}}}{\delta_2}
+ \sqrt{\left(\frac{\mu_{G^{1,1}}}{\delta_1} + \frac{\mu_{G^{2,2}}}{\delta_2}\right)^2
+ 4 \frac{\mu_{G^{1,2}}}{\delta_2}\frac{\mu_{G^{2,1}}}{\delta_1}} \right)
< 1.
\end{eqnarray} 
\end{remark}

The following lemma shows a necessary condition of the existence of the limiting
distribution. It also indicates
that the limiting distribution is independent from the initial condition
with (\ref{eqn: Laplace_L_m_T_forward}).
\begin{lemma}\label{lem: lfunction_asympt_finiteSystem}
For any $i=1,\ldots,m$, and $\left(v_1,\ldots, v_m \right)\in \mathbb{R}^m_+$,
\begin{eqnarray*}
\lim_{t\rr \infty}l_i(t) = 0.
\end{eqnarray*}
\end{lemma}
\begin{proof}
See Section~\ref{sec: proof_lem_lfunction_asympt_finiteSystem}.
\end{proof}

Then the Laplace transform (\ref{eqn: Laplace_L_m_T_forward}) of the limiting
distribution of the univariate
finite system $(\Lambda^{(1)},\ldots,\Lambda^{(m)})$ at any $(v_1,\ldots,v_m)
\in \R^m_{+}$ becomes
\begin{eqnarray}\label{eqn: Lap_truncated_asymptotic}
\hat{\pi}^m_A(v_1,\ldots, v_m) &:=& \lim_{T\rr\infty}\E_0\left[ e^{-v_1\Lambda^{(1)}_T
-
\cdots - v_m\Lambda^{(m)}_T } \right]  \nonumber\\
&=&\exp\left(-\rho_1 \int_0^\infty \left[1-\hat{h}_1(l_{2n}(t))\right]dt
-\rho_2 \int_0^\infty \left[1-\hat{h}_2(l_{2n-1}(t))\right]dt \right).\nonumber\\
\end{eqnarray}

In the following theorem, we provide the existence condition for the limiting
distributions.
\begin{theorem}[Existence of Limiting Distributions]\label{thm: LimitingDistributionAll}
\quad \\
(1) Under Condition $(C1)$, as $t \rr \infty$, the limiting distributions
$\pi^m_A$ of $\left(\Lambda^{(1)},\ldots, \Lambda^{(m)} \right)$ and $\mu^n_A$
of  $\left(\lambda^{1,n}, \lambda^{2,n} \right)$ exist.\\
(2) Under Condition $(C1)$ and $(C2)$, as $t\rr \infty$, the limiting distribution
$\mu^*_A$ of $\left(\lambda^1,\lambda^2\right)$ exits.
\end{theorem}

\begin{proof}
(1)
Since for any $k\in \mathbb{N}$,
\begin{eqnarray}\label{eqn: OneMinusHatH}
1-\hat{h}_1\left(l_{2k}(t) \right)  = \int_t^\infty d \hat{h} _1\left(l_{2k}(u)
\right) = \int_t^\infty \hat{h}'_1(l_{2k}(u)) \dot{l}_{2k}(u)du \leq \mu_{H^1}
l_{2k}(t).
\end{eqnarray}

Hence $\int_0^\infty \left[1-\hat{h}_1\left(l_{2k}(t) \right) \right]dt \leq
\mu_{H^1}\int_0^\infty l_{2k}(t) dt$ and (\ref{eqn: Lap_truncated_asymptotic})
becomes 
\begin{eqnarray}\label{eqn: Laplace_Limiting_Bound}
\hat{\pi}^m_A (v_1,\ldots,v_m) \geq \exp\left(-\rho_1 \mu_{H^1} \int_0^\infty
l_{2n}(t) dt -\rho_2 \mu_{H^2} \int_0^\infty l_{2n-1}(t)
dt\right).
\end{eqnarray}

In order to show the existence of the limiting distribution $\pi^m_A$, it
is sufficient to show that 
$\begin{bmatrix}
\int_0^\infty l_{2n-1}(t)dt\\ 
\int_0^\infty l_{2n}(t)dt
\end{bmatrix}
$ is finite and thus the process does not explode as $t\rr\infty$.

From (\ref{eqn: l_function}), for all $j = 1,\ldots,2n$, and $t\geq 0$, the
function $l_j(t)$ is increasing with the initial value $v_{2n+1-j}$. We construct
a sequence of functions $\{L_j(t)\}_{j=1}^{2n}$ that is the solution to the
forward ODE system (\ref{eqn: ForwardODE}) with the initial value
\begin{equation}\label{eqn: InitialValue_L_dominant}
L_{2k-1}(0) = v_2^* = \max_{i=1,\ldots,n} v_{2i},\quad L_{2k}(0) = v^*_1=
\max_{i=1,\ldots,n} v_{2i-1} \quad \text{for $k=1,\ldots,n$}.
\end{equation}
Therefore $l_{j}(t)\leq L_{j}(t)$ for $i=1,\ldots,2n$,  and $\begin{bmatrix}
\int_0^\infty l_{2n-1}(t)dt\\ 
\int_0^\infty l_{2n}(t)dt
\end{bmatrix}
\leq 
\begin{bmatrix}
\int_0^\infty L_{2n-1}(t)dt\\ 
\int_0^\infty L_{2n}(t)dt
\end{bmatrix}
$. Then it is sufficient to show that $\begin{bmatrix}
\int_0^\infty L_{2n-1}(t)dt\\ 
\int_0^\infty L_{2n}(t)dt
\end{bmatrix}
<\infty$.

With the same initial value as in (\ref{eqn: InitialValue_L_dominant}), from
the explicit recursive solution (\ref{eqn: ForwardODE}), one can easily check
by induction that for each $t\geq 0$, $L_{2k-1}(t)$ and $L_{2k}(t)$ are increasing
with $k$. Hence we can define
non-negative distance functions 
\begin{eqnarray*}
k =1: \quad &&  d^{(1)}_1(t) := L_{1}(t), \quad d^{(2)}_1(t) := L_{2}(t).\\
k \geq 2: \quad && d^{(1)}_k(t) := L_{2k-1}(t) - L_{2k-3}(t), \quad d^{(2)}_k(t)
:= L_{2k}(t) - L_{2k-2}(t).
\end{eqnarray*}
The following inequalities hold and are proven in Section~\ref{sec: proof_distanceFunction}:
\begin{equation}\label{eqn: distance_function}
\begin{aligned}
d^{(1)}_{k+1}(t) &\leq  e^{-\delta_2 t} \int_0^t e^{\delta_2 s} \left[\mu_{G^{2,2}}
d^{(1)}_k(s) + \mu_{G^{1,2}}d^{(2)}_k(s) \right] ds, \\
d^{(2)}_{k+1}(t) &\leq  e^{-\delta_1 t} \int_0^t e^{\delta_1 s} \left[ \mu_{G^{2,1}}
d^{(1)}_k(s)+ \mu_{G^{1,1}}d^{(2)}_k(s) \right] ds. 
\end{aligned}
\end{equation}

\begin{eqnarray*}
\int_0^\infty d^{(1)}_{i+1}(t) dt &\leq & \int_{t=0}^\infty e^{-\delta_2
t} \int_{s=0}^t e^{\delta_2 s} \left( \mu_{G^{2,2}} d_i^{(1)}(s) + \mu_{G^{1,2}}
d_i^{(2)}(s) \right)ds \\
&=&  \int_{s=0}^\infty \left(\int_{t=s}^\infty e^{-\delta_2 t}dt \right)
 e^{\delta_2 s} \left( \mu_{G^{2,2}} d_i^{(1)}(s) + \mu_{G^{1,2}} d_i^{(2)}(s)
\right)ds\\
&=& \frac{\mu_{G^{2,2}}}{\delta_2}\int_0^\infty d^{(1)}_i(s)ds + \frac{\mu_{G^{1,2}}}{\delta_2}\int_0^\infty
d^{(2)}_i(s)ds. 
\end{eqnarray*}
Similarly, 
\begin{eqnarray*}
\int_0^\infty d^{(2)}_{i+1}(t) dt &\leq &\frac{\mu_{G^{2,1}}}{\delta_1}\int_0^\infty
d^{(1)}_i(s)ds + \frac{\mu_{G^{1,1}}}{\delta_1}\int_0^\infty d^{(2)}_i(s)ds.
\end{eqnarray*}

i.e.
\begin{eqnarray*}
\begin{bmatrix}
\int_0^\infty d^{(1)}_{i+1}(t)dt\\ 
\int_0^\infty d^{(2)}_{i+1}(t)dt
\end{bmatrix} 
\leq A 
\begin{bmatrix}
\int_0^\infty d^{(1)}_{i}(t)dt\\ 
\int_0^\infty d^{(2)}_{i}(t)dt
\end{bmatrix},\quad \text{with } A := \begin{bmatrix}
\frac{\mu_{G^{2,2}}}{\delta_2} & \frac{\mu_{G^{1,2}}}{\delta_2}\\ 
\frac{\mu_{G^{2,1}}}{\delta_1} &
\frac{\mu_{G^{1,1}}}{\delta_1}
\end{bmatrix}.
\end{eqnarray*}

Iteratively, we obtain for $i\geq 1$,
\begin{eqnarray*}
\begin{bmatrix}
\int_0^\infty d^{(1)}_{i}(t)dt\\ 
\int_0^\infty d^{(2)}_{i}(t)dt
\end{bmatrix} 
\leq A^{i-1} 
\begin{bmatrix}
\int_0^\infty d^{(1)}_{1}(t)dt\\ 
\int_0^\infty d^{(2)}_{1}(t)dt
\end{bmatrix} 
= A^{i-1} 
\begin{bmatrix}
\frac{v_2^{*}}{\delta_2}\\ 
\frac{v_1^{*}}{\delta_1}
\end{bmatrix}. 
\end{eqnarray*}

Denote the spectral radius of $A$ as $\rho$. From the matrix theory,  for
any $\epsilon>0$, and denote $\td{\rho}:=\rho+\epsilon$, there exists a
norm $\|\cdot\|$, such that $ \|A\|\leq \td{\rho}$. Then, for any $i\geq
1$, 
$$\|A^{i}\|\leq \|A\|^{i} \leq   \td{\rho}^i.$$
Moreover, take the Euclidean norm and due to the equivalence of norm,
there exists a constant $C>0$, such that  
\begin{eqnarray*}
\|A^{i}\|_2 \leq C \|A^{i}\| \leq C \td{\rho}^{i}.
\end{eqnarray*}

By definition, $L_{2n-1}(t)=  \sum_{i=1}^n d^{(1)}_i(t)$ and $
L_{2n}(t)  = \sum_{i=1}^n
d^{(2)}_i(t)$,
then
\begin{eqnarray*}
\begin{bmatrix}
\int_0^\infty L_{2n-1}(t)dt\\ 
\int_0^\infty L_{2n}(t)dt
\end{bmatrix} = 
\sum_{i=1}^n
\begin{bmatrix}
\int_0^\infty d^{(1)}_i(t)dt\\ 
\int_0^\infty d^{(2)}_i(t)dt
\end{bmatrix}
\leq \left( \sum_{i=1}^n A^{i-1}\right) \begin{bmatrix}
\frac{v_2^*}{\delta_2}\\ 
\frac{v_1^*}{\delta_1}
\end{bmatrix}.
\end{eqnarray*}
Denote $
\td{L}_n:=\left\| 
\begin{bmatrix}
\int_0^\infty L_{2n-1}(t)dt\\ 
\int_0^\infty L_{2n}(t)dt
\end{bmatrix} 
\right\|_2 $, then
\begin{eqnarray*}
\td{L}_n \leq \sum_{i=1}^n \| A^{i-1} \|_2 \left\|   \begin{bmatrix}
\frac{v_2^{*}}{\delta_2}\\ 
\frac{v_1^*}{\delta_1}
\end{bmatrix}
\right\|_2
\leq  C \frac{1-\td{\rho}^n}{1-\td{\rho}} \sqrt{\left(\frac{v_2^*}{\delta_2}\right)^2
+ \left(\frac{v_1^*}{\delta_1}\right)^2} < \infty.
\end{eqnarray*}

Hence $\int_0^\infty L_{2n-1}(t)dt \leq \td{L}_n < \infty$ and $\int_0^\infty
L_{2n}(t)dt \leq\td{L}_n < \infty$, which indicates $\pi^m_A$ exists.\\

The existence of the limiting distribution $\mu^n_A$ of $(\lambda^{1,n},\lambda^{2,n})$
is indicated from the analysis above. Indeed, by taking  $v_{2i-1} = v_1$
and $v_{2i} = v_2$ for $i=1,\ldots,n$, then the Laplace transform in (\ref{eqn:
Lap_truncated_asymptotic}) becomes
\begin{equation}\label{eqn: Lapalace_truncated_v2}
\begin{aligned}
\hat{\mu}^n_A(v_1,v_2) &:= \lim_{T\rightarrow \infty}\E_0\left[ e^{-v_1\lambda^{1,n}_T
 - v_2\lambda^{2,n}_T } \right]\\
&=\exp\left(-\rho_1 \int_0^\infty \left[1-\hat{h}_1(l_{2n}(t))\right]dt -\rho_2
\int_0^\infty \left[1-\hat{h}_2(l_{2n-1}(t))\right]dt \right),
\end{aligned}
\end{equation} 
where $l_{2n-1}(t)$, $l_{2n}(t)$ are from the solution of the ODE system
(\ref{eqn: ForwardODE}) with initial values $l_{2i-1}
(0) = v_1 $ and $l_{2i}(0) \equiv v_2 $ for $i=1,\ldots, n$.

In this case, $l_{j}(t) = L_j(t)$ for $j=1,\ldots,m$, and therefore the limiting
distribution $\mu^n_A$ exists.\\

(2) We explore the existence condition of the limiting distribution $\mu^*_A$
of $(\lambda^1,\lambda^2)$ using the convergence from $\mu^n_A$.

Note that for the Laplace transform $\hat{\mu}^n_A$ in (\ref{eqn: Lapalace_truncated_v2}),
$l_{2n}(t)$ and $l_{2n-1}(t)$ are from the explicit solution to (\ref{eqn:
ForwardODE}). That is, for $k=0,\ldots, n-1$,   
\begin{equation}
\begin{aligned}
l_{1}(t) &= v_2 e^{-\delta_2 t},\quad l_{2}(t) = v_1 e^{-\delta_1 t}\\
l_{2k+1}(t) &= v_{2}e^{-\delta_2 t} + e^{-\delta_2 t}\int_0^t e^{\delta_2
s}\left[ 1-\hat{g}_{1,2}(l_{2k}(s))+1-\hat{g}_{2,2}(l_{2k-1}(s))\right] ds\\
l_{2k+2}(t) &= v_{1}e^{-\delta_1 t} + e^{-\delta_1 t}\int_0^t e^{\delta_1
s}\left[ 1-\hat{g}_{1,1}(l_{2k}(s))+1-\hat{g}_{2,1}(l_{2k-1}(s))\right] ds.
\end{aligned}
\end{equation}

Note that $l_{2k-1}(t)$ and $l_{2k}(t)$ are increasing functions of $k$ for
all $k$ and $t\geq 0$, hence by the monotone convergence theorem, $(l_{2n-1}(t)
,l_{2n}(t))$ converges
to a limit $(l_1^*(t),l_2^*(t))$.

The Laplace transform of the limiting distribution of $(\lambda^1_t,\lambda^2_t)$
is
\begin{eqnarray}\label{eqn: Laplace_Limiting_BDCP_Intensity}
\hat{\mu}_A^*(v_1,v_2) &=& \lim_{n \rr \infty}\hat{\mu}^n_A(v_1,v_2)
= \exp\left(-\rho_1 \int_0^\infty \left[1-\hat{h}_1(l_2^*(t))\right]dt -\rho_2
\int_0^\infty \left[1-\hat{h}_2(l_1^*(t))\right]dt \right).\nonumber\\
\end{eqnarray}
To show $\mu^*_A$ is non-degenerate, following the same argument as in (\ref{eqn:
OneMinusHatH}), it is sufficient to have$
\begin{bmatrix}
\int_0^\infty l_{1}^*(t)dt\\ 
\int_0^\infty l_{2}^*(t)dt
\end{bmatrix} 
<\infty.
$ 

Under $(C2)$ in Condition \ref{asm: stationaryCondition} and the matrix theory,
 take $0<\epsilon<\frac{1-\rho}{2}$,
there exists a
norm $\|\cdot\|$, such that $ \|A\|\leq \td{\rho} = \rho + \epsilon < 1$,
then from the first part of the proof, 
\begin{eqnarray*}
\td{L}_n 
\leq  C \frac{1-\td{\rho}^n}{1-\td{\rho}} \sqrt{\left(\frac{v_2}{\delta_2}\right)^2
+ \left(\frac{v_1}{\delta_1}\right)^2} <  C \frac{1}{1-\td{\rho}} \sqrt{\left(\frac{v_2}{\delta_2}\right)^2
+ \left(\frac{v_1}{\delta_1}\right)^2}.
\end{eqnarray*}
Therefore,
\begin{eqnarray*}
\begin{bmatrix}
\int_0^\infty l_{1}^*(t)dt\\ 
\int_0^\infty l_{2}^*(t)dt
\end{bmatrix} = 
\lim_{n\rr \infty} 
\begin{bmatrix}
\int_0^\infty l_{2n-1}(t)dt\\ 
\int_0^\infty l_{2n}(t)dt
\end{bmatrix} 
<\infty.
\end{eqnarray*}

\end{proof}

\section{Stationary Distribution}\label{sec: StationaryDistribution}
 The limiting distributions of the finite system and BDCP exist by Theorem~\ref{thm:
LimitingDistributionAll}. In this part, we show the equivalence between the
stationary distribution and the limiting distribution.

First, a stationarity condition for the finite system  $\left(\Lambda^{(1)}_t,\ldots,
\Lambda^{(m)}_t\right)$
is provided. {\color{blue}  }
\begin{lemma}[Stationary condition equation for the finite system]
\label{thm: statCondTruncSys}\quad\\
Distribution $\pi^m_S$ is a stationary distribution of $(\Lambda^{(1)}_t,\ldots,
\Lambda^{(m)}_t)$ if and only if the Laplace transform $\hat{\pi}^m_S$ at
any $(v_1,\ldots,v_m)\in \R^m_{+}$ satisfies
\begin{eqnarray*}
0 &=& -\sum_{k=1}^m \delta_k v_k \hat{\pi}^m(v_1,\ldots, v_m) + \rho_1(\hat{h}(v_1)-1)
+\rho_2(\hat{h}(v_2)-1) \nonumber \\
&& + \sum_{k=1}^{n-1} \frac{\partial \hat{\pi}^m(v_1,\ldots,v_m)}{\partial
v_{2k-1}} \left[(1-\hat{g}_{1,1}(v_{2k+1})) + (1-\hat{g}_{2,1}(v_{2k+2}))
\right] 
\\
&& + \sum_{k=1}^{n-1} \frac{\partial \hat{\pi}^m(v_1,\ldots,v_m)}{\partial
v_{2k}} \left[(1-\hat{g}_{1,2}(v_{2k+1})) + (1-\hat{g}_{2,2}(v_{2k+2})) \right].
\end{eqnarray*}

Equivalently, in terms of ODE system (\ref{eqn:
ForwardODE}) as
\begin{eqnarray}\label{eqn: StatEq_v1}
0 = \sum_{k=1}^{n} \dot{l}_{2(n-k)+2}(0) \frac{\partial\hat{\pi}^m_S}{\partial
v_{2k-1}}  +  \dot{l}_{2(n-k)+1}(0) \frac{\partial\hat{\pi}^m_S}{\partial
v_{2k}}  - \rho_1 (1-\hat{h}_1({v_1}))\hat{\pi}^m_S - \rho_2 (1-\hat{h}_2({v_2}))\hat{\pi}^m_S.\nonumber\\
\end{eqnarray}

\end{lemma}

\begin{proof}
The proof is based on Markov theory as in (\ref{eqn: Stationarity_Markov})
and
details can be found in Section~\ref{sec: ProofThm_statCondTruncSys}.
\end{proof}

The following theorem states the equivalence between the limiting and stationary
distribution.
\begin{theorem}[Stationarity of the finite system]\label{prop: AsyEqStat}
\quad\\
For any system index $m=2n$ for $\left(\Lambda^{(1)}_t,\ldots,
\Lambda^{(m)}_t\right)$, if there exits a limiting distribution $\pi^m_A$,
then there exits a unique stationary distribution $\pi^m_S$ and $\pi^m_S
\stackrel{d}{=} \pi^m_A$. 
\end{theorem}

\begin{proof}
For the existence, it is sufficient to show that $\hat{\pi}^m:=\hat{\pi}^m_A$
satisfies the condition equation (\ref{eqn: StatEq_v1}). The uniqueness of
such $\pi^m_S$ follows from the uniqueness of $\pi^m_A$.

Since $\lim_{t\rr\infty}l_{2n}(t)
= \lim_{t\rr \infty} l_{2n-1}(t) = 0$ from Lemma~\ref{lem: lfunction_asympt_finiteSystem},
\begin{eqnarray*}
\frac{\partial\hat{\pi}^m}{\partial v_{2k-1}} &=& \hat{\pi}^m\left[\rho_1
\int_0^\infty   \hat{h}'_1(l_{2n}(t))\frac{\partial l_{2n}(t)}{\partial v_{2k-1}}
dt +\rho_2 \int_0^\infty \hat{h}'_2(l_{2n-1}(t))\frac{\partial l_{2n-1}(t)}{\partial
v_{2k-1}} dt \right]\\
\frac{\partial\hat{\pi}^m}{\partial v_{2k}} &=& \hat{\pi}^m \left[\rho_1
\int_0^\infty \hat{h}'_1(l_{2n}(t))\frac{\partial l_{2n}(t)}{\partial v_{2k}}
dt +\rho_2 \int_0^\infty \hat{h}'_2(l_{2n-1}(t))\frac{\partial l_{2n-1}(t)}{\partial
v_{2k}} dt \right]\\
1-\hat{h}_1({v_1}) &=& \int_0^\infty \hat{h}'_1(l_{2n}(t))\dot{l}_{2n}(t)
dt\\
1-\hat{h}_2({v_2}) &=& \int_0^\infty \hat{h}'_2(l_{2n-1}(t))\dot{l}_{2n-1}(t)
dt.
\end{eqnarray*}

Then, the stationarity equation (\ref{eqn: StatEq_v1}) becomes
\begin{eqnarray*}
0 &=& \rho_1 \int_0^\infty \hat{h}'_1(l_{2n}(t)) \left[\sum_{k=1}^n \left(\dot{l}_{2(n-k)+2}(0)\frac{\partial
l_{2n}(t)}{\partial v_{2k-1}} + \dot{l}_{2(n-k)+1}(0)\frac{\partial
l_{2n}(t)}{\partial v_{2k}} \right) -\dot{l}_{2n}(t)\right]dt \\
&& + \rho_2\int_0^\infty \hat{h}'_2(l_{2n-1}(t)) \left[\sum_{k=1}^n \left(
\dot{l}_{2(n-k)+2}(0)\frac{\partial
l_{2n-1}(t)}{\partial v_{2k-1}} +  \dot{l}_{2(n-k)+1}(0)\frac{\partial
l_{2n-1}(t)}{\partial v_{2k}}\right)  -\dot{l}_{2n-1}(t)\right] dt.
\end{eqnarray*}

Since the functions $l_k$ is independent from the choice of $\hat{h}_i$ and
$\rho_i$ for $i=1,2$, for system index $m=2n$, re-denote $l_k(\cdot)$ as
$l^{2n}_k(\cdot)$, then it is sufficient to show:
\begin{equation}\label{eqn: StatEq_v2}
  \begin{aligned}
\sum_{k=1}^n\left( \dot{l}^{2n}_{2(n-k)+2}(0)\frac{\partial l^{2n}_{2n}(t)}{\partial
v_{2k-1}} + \dot{l}^{2n}_{2(n-k)+1}(0)\frac{\partial l^{2n}_{2n}(t)}{\partial
v_{2k}} \right) -\dot{l}^{2n}_{2n}(t) &=& 0, \\
\sum_{k=1}^n \left( \dot{l}^{2n}_{2(n-k)+2}(0)\frac{\partial l^{2n}_{2n-1}(t)}{\partial
v_{2k-1}} +  \dot{l}^{2n}_{2(n-k)+1}(0)\frac{\partial l^{2n}_{2n-1}(t)}{\partial
v_{2k}} \right) - \dot{l}^{2n}_{2n-1}(t) &=& 0.
  \end{aligned}
\end{equation}

By observing the self-similarity in the system structure, we prove (\ref{eqn:
StatEq_v2}) using the induction with respect to the system index $m = 2n$.
\begin{enumerate}
\item[(1)] For $n=1$, it is easy to observe that $l^2_1(t) = v_2 e^{-\delta_2
t}$ and $l^2_2(t) = v_1 e^{-\delta_1 t}$ satisfies (\ref{eqn: StatEq_v2}).

\item[(2)] Assume that $m=2n$ satisfies (\ref{eqn: StatEq_v2}), we show that
for $m=2(n+1)$ also satisfies (\ref{eqn: StatEq_v2}). One needs to show the
first equation as follows and the second follows in the same way. I.e.,
\begin{eqnarray}\label{eqn: induction_l_n+1}
\sum_{k=1}^{n+1}\left[ \dot{l}^{2(n+1)}_{2(n+1-k)+2}(0)\frac{\partial l^{2(n+1)}_{2n+2}(t)}{\partial
v_{2k-1}} + \dot{l}^{2(n+1)}_{2(n+1-k)+1}(0)\frac{\partial l^{2(n+1)}_{2n+2}(t)}{\partial
v_{2k}} \right] -\dot{l}^{2(n+1)}_{2n+2}(t) &=& 0, \nonumber \\
\end{eqnarray}
where $l_{2(n+1)-i+1}^{2(n+1)}(0) = v_i$ for $i=1,\ldots,2n$. 

Note that from the ODE system and the recursive solution, we have
$$l^{2(n+1)}_{2(n+1)}(0) = v_1,\quad \frac{\partial l^{2(n+1)}_{2(n+1)}(t)}{\partial
v_1} = e^{-\delta_1 t},\quad \frac{\partial l^{2(n+1)}_{2(n+1)}(t)}{\partial
v_2} = 0 .$$ Then, the $k=1$ term in (\ref{eqn: induction_l_n+1}) becomes
$ \dot{l}^{2(n+1)}_{2n+2}(0)e^{-\delta_1 t}$. Hence, we need to show that
$$ \sum_{k=2}^{n+1} \left[\dot{l}^{2(n+1)}_{2(n+1-k)+2}(0)\frac{\partial
l^{2(n+1)}_{2n+2}(t)}{\partial
v_{2k-1}} + \dot{l}^{2(n+1)}_{2(n+1-k)+1}(0)\frac{\partial l^{2(n+1)}_{2n+2}(t)}{\partial
v_{2k}}  \right] = \dot{l}^{2(n+1)}_{2n+2}(t) -e^{-\delta_1
t} \dot{l}^{2(n+1)}_{2n+2}(0).$$

As (\ref{eqn: StatEq_v2}) holds for $m=2n$ for
all $(v_1,\ldots,v_{2n}) \in
\R^m_+$, one can construct functions $\{L^{2n}_i(\cdot)\}_{i=1}^{2n}$, such
that they satisfies  (\ref{eqn: StatEq_v2})  with initial values $$ (\td{v}_1,\ldots,
\td{v}_{2n}) = (v_3,\ldots,v_{2n+2}).$$ 
 
Hence, \begin{equation}\label{eqn: L_2n}
\begin{aligned}
\sum_{k=1}^n\left(\dot{L}^{2n}_{2(n-k)+2}(0)\frac{\partial L^{2n}_{2n}(t)}{\partial
\td{v}_{2k-1}} + \dot{L}^{2n}_{2(n-k)+1}(0)\frac{\partial L^{2n}_{2n}(t)}{\partial
\td{v}_{2k}} \right) -\dot{L}^{2n}_{2n}(t) &=& 0\\
\sum_{k=1}^n\left(\dot{L}^{2n}_{2(n-k)+2}(0)\frac{\partial L^{2n}_{2n-1}(t)}{\partial
\td{v}_{2k-1}} + \dot{L}^{2n}_{2(n-k)+1}(0)\frac{\partial L^{2n}_{2n-1}(t)}{\partial
\td{v}_{2k}} \right) -\dot{L}^{2n}_{2n-1}(t) &=& 0.
\end{aligned}
\end{equation}
with $L_{2n-i+1}^{2n}(0) = \td{v}_i = v_{i+2}$ for $i=1,\ldots,2n$. Especially
$L_1^{2n}(0)=v_{2n+2}$ and
$L_2^{2n}(0)=v_{2n+1}$.

By construction construction,  for $k=1,\ldots, n$, $t\geq
0$,
\begin{equation*}
 l^{2(n+1)}_{2k-1}(t) =  L^{2n}_{2k-1}(t),\quad l^{2(n+1)}_{2k}(t) =  L^{2n}_{2k}(t),
\end{equation*}
and
\begin{equation*}
\frac{\partial l^{2(n+1)}_{2n}(t)}{\partial v_{2k-1}} = \frac{\partial L^{2n}_{2n}(t)}{\partial
v_{2k-1}} = \frac{\partial L^{2n}_{2n}(t)}{\partial \td{v}_{2k-3}},\quad
\frac{\partial l^{2(n+1)}_{2n}(t)}{\partial v_{2k}} = \frac{\partial L^{2n}_{2n}(t)}{\partial
v_{2k}} = \frac{\partial L^{2n}_{2n}(t)}{\partial \td{v}_{2k-2}}.  
\end{equation*}

For terms with $k\geq 2$ in in (\ref{eqn: induction_l_n+1}),
\begin{eqnarray}
&&\sum_{k=2}^{n+1} \left[\dot{l}^{2(n+1)}_{2(n+1-k)+2}(0)\frac{\partial l^{2(n+1)}_{2n+2}(t)}{\partial
v_{2k-1}} + \dot{l}^{2(n+1)}_{2(n+1-k)+1}(0)\frac{\partial l^{2(n+1)}_{2n+2}(t)}{\partial
v_{2k}}  \right] \nonumber \\
&=& -  \int_0^t  e^{-\delta_1 (t-s)} \hat{g}'_{1,1}\left(L^{2n}_{2n}(s) \right)\sum_{k=2}^{n+1}\left[
\dot{L}^{2n}_{2(n+1-k)+2}(0) \frac{\partial L^{2n}_{2n}(t)}{\partial \td{v}_{2k-3}}
+ \dot{L}^{2n}_{2(n+1-k)+1}(0) \frac{\partial L^{2n}_{2n}(t)}{\partial \td{v}_{2k-2}}
\right]ds\nonumber\\
&& - \int_0^t  e^{-\delta_1 (t-s)} \hat{g}'_{2,1}\left(L^{2n}_{2n-1}(s) \right)\sum_{k=2}^{n+1}\left[
\dot{L}^{2n}_{2(n+1-k)+2}(0) \frac{\partial L^{2n}_{2n-1}(t)}{\partial \td{v}_{2k-3}}
+ \dot{L}^{2n}_{2(n+1-k)+1}(0) \frac{\partial L^{2n}_{2n-1}(t)}{\partial
\td{v}_{2k-2}} \right]ds \nonumber\\
&=& 
- \int_0^t  e^{-\delta_1 (t-s)} \hat{g}'_{1,1}\left(L^{2n}_{2n}(s) \right)\sum_{k=1}^{n}\left[
\dot{L}^{2n}_{2(n-k)+2}(0) \frac{\partial L^{2n}_{2n}(t)}{\partial \td{v}_{2k-1}}
+ \dot{L}^{2n}_{2(n-k)+1}(0) \frac{\partial L^{2n}_{2n}(t)}{\partial \td{v}_{2k}}
\right]ds\nonumber\\
&&- \int_0^t  e^{-\delta_1 (t-s)} \hat{g}'_{2,1}\left(L^{2n}_{2n-1}(s) \right)\sum_{k=1}^{n}\left[
\dot{L}^{2n}_{2(n-k)+2}(0) \frac{\partial L^{2n}_{2n-1}(t)}{\partial \td{v}_{2k-1}}
+ \dot{L}^{2n}_{2(n-k)+1}(0) \frac{\partial L^{2n}_{2n-1}(t)}{\partial \td{v}_{2k}}
\right]ds. \nonumber\\ \label{eqn: InductionProofEqn1}
\end{eqnarray}

By (\ref{eqn: L_2n}), (\ref{eqn: InductionProofEqn1}) becomes 
\begin{eqnarray*}\label{eqn: Induction_intermediate}
&& - \int_0^t  e^{-\delta_1 (t-s)} \hat{g}'_{1,1}\left(L^{2n}_{2n}(s) \right)
\dot{L}^{2n}_{2n}(s)ds  
- \int_0^t  e^{-\delta_1 (t-s)} \hat{g}'_{2,1}\left(L^{2n}_{2n-1}(s) \right)
\dot{L}^{2n}_{2n-1}(s)ds \nonumber  \\
&=& \int_0^t  e^{-\delta_1 (t-s)} \frac{\partial}{\partial s} \left[ 1 -
\hat{g}_{1,1}\left({L}^{2n}_{2n}(s)\right) +1 - \hat{g}_{2,1}\left({L}^{2n}_{2n-1}(s)
\right) \right] ds\nonumber\\
&=& \int_0^t  e^{-\delta_1 (t-s)} \frac{\partial}{\partial s} \left[ 1 -
\hat{g}_{1,1}\left({l}^{2(n+1)}_{2n}(s)\right) + 1 - \hat{g}_{2,1}\left({l}^{2(n+1)}_{2n-1}(s)
\right) \right] ds\nonumber\\
&\stackrel{(\ref{eqn: ForwardODE})}{=}&  \int_0^t  e^{-\delta_1 (t-s)} \frac{\partial}{\partial
s} \left[\dot{l}^{2(n+1)}_{2n+2}(s) + \delta_1 l^{2(n+1)}_{2n+2}(s)\right]
ds
\end{eqnarray*}

Denote $$F(s):= e^{\delta_1 s}l^{2(n+1)}_{2n+2}(s),$$ 
then $\dot{F}(s) = e^{\delta_1
s}\left(\dot{l}^{2(n+1)}_{2n+2}(s) + \delta_1 l^{2(n+1)}_{2n+2}(s)\right)$,
and $\dot{F}(0) =\dot{l}^{2(n+1)}_{2n+2}(0) + \delta_1 v_1$.
\begin{eqnarray*}
(\ref{eqn: InductionProofEqn1})
&=& e^{-\delta_1 t}\int_0^t e^{\delta_1 s}d\left(e^{-\delta_1  s}\dot{F}(s)\right)
\\
&=& e^{-\delta_1 t} \left( \dot{F}(t) - \dot{F}(0) \right) - \delta_1 e^{-\delta_1
t} \left(F(t) - F(0) \right)\\
&=& \left(\dot{l}^{2(n+1)}_{2n+2}(t) + \delta_1
l^{2(n+1)}_{2n+2}(t)\right)-e^{-\delta_1 t}\left(\dot{l}^{2(n+1)}_{2n+2}(0)
+ \delta_1 v_1\right) - \delta_1 e^{-\delta_1 t}\left( e^{\delta_1
t}l^{2(n+1)}_{2n+2}(t) - v_1\right)\\
&=& \dot{l}^{2(n+1)}_{2n+2}(t) -e^{-\delta_1 t}\dot{l}^{2(n+1)}_{2n+2}(0).
\end{eqnarray*}
\end{enumerate}

\end{proof}

Hence, by Theorem~\ref{prop: AsyEqStat} and Theorem~\ref{thm: LimitingDistributionAll},
we can conclude that for the finite system $\left(\Lambda^{(1)}_t,\ldots,
\Lambda^{(m)}_t\right)$, there exists a unique stationary distribution that
is equal to the limiting distribution.

\begin{corollary}\label{sec: Stationarity_BDCP_lambdaN}
There exists a unique stationary distribution $\mu^n_S$ of $\left(\lambda^{1,n}_t,\lambda^{2,n}_t\right)$,
and it is equal to the limiting distribution $\mu^n_A$ with the Laplace transform
(\ref{eqn: Lapalace_truncated_v2}).

\end{corollary}
\begin{proof}

The joint distribution of $\left(\lambda^{1,n}_t,\lambda^{2,n}_t\right)$
at any $t\geq 0$, $k\geq 0$ and $0\leq s_1\leq \ldots \leq s_k$ is
\begin{eqnarray*}
&&\Prob\left(\lambda^{1,n}_{t+s_1}\leq x_1^1,\lambda^{2,n}_{t+s_1}\leq x_1^2;\ldots
; \lambda^{1,n}_{t+s_k}\leq x_k^1,\lambda^{2,n}_{t+s_k}\leq x_k^2 \right)
\nonumber\\
&=& \Prob\left(\sum_{i=1}^n \lambda^{1,(i)}_{t+s_1}\leq x_1^1 , \sum_{i=1}^n
\lambda^{2,(i)}_{t+s_1}\leq x_1^2; \ldots ;  \sum_{i=1}^n \lambda^{1,(i)}_{t+s_k}\leq
x_k^1 , \sum_{i=1}^n \lambda^{2,(i)}_{t+s_k}\leq x_k^2 \right)\\
&=& \int_{D_1^1}\int_{D_1^2} \cdots \int_{D_k^1}\int_{D_k^2} d\Prob\left(\lambda^{1,(1)}_{t+s_1}\leq
z^{1,(1)}_1,\lambda^{2,(1)}_{t+s_1}\leq z^{2,(1)}_1,\ldots, \lambda^{1,(n)}_{t+s_k}
\leq z^{1,(n)}_k, \lambda^{2,(n)}_{t+s_k} \leq z^{2,(n)}_k \right),  
\end{eqnarray*}
where for $j=1,\ldots,k, j' = 1,2$, $D_{j}^{j'} = \left\{\left(z^{j',(1)}_j,\ldots,z^{j',(n)}_j\right)\in
\mathbb{R}^{n}: \sum_{i=1}^n z^{j',(i)}_j \leq x^{j'}_j\right\}$.

By Theorem~\ref{prop: AsyEqStat}, take the distribution of $\left(\lambda^{1,(1)}_t,\lambda^{2,(1)}_t,\ldots,\lambda^{1,(n)}_t,\lambda^{2,(n)}_t\right)=\left(\Lambda^{(1)},\ldots,\Lambda^{(m)}\right)$
as the unique stationary distribution $\pi^m_S$, then the joint
distribution above is independent of $t$. Hence from the equation above,
the distribution $\left(\lambda^{1,n}_{t+s_1},\lambda^{2,n}_{t+s_1};\ldots
; \lambda^{1,n}_{t+s_k},\lambda^{2,n}_{t+s_k} \right)$ is also independent
of $t$. Therefore by definition $\left(\lambda^{1,n}_t,\lambda^{2,n}_t\right)$
is a stationary process. 

Since the limiting distribution exits and independent from the initial condition,
then $\mu^n_S \stackrel{d}{=} \mu^n_A$ and the uniqueness follows. 
\end{proof}

Now we present  the existence and uniqueness
of stationary distribution for $(\lambda^1_t,\lambda^2_t)$ which is the main
result of this paper.

\begin{theorem}[Existence of Stationary Distribution]\label{thm: stationarity_BDCP}
Under (C1) and (C2), there exists a unique stationary distribution $\mu^*_S$
for the BDCP intensity $\left(\lambda^1_t,\lambda^2_t
\right)$, and moreover $\mu^*_S \stackrel{d}{=} \mu^*_A$.
\end{theorem}
\begin{proof}
Let $(\lambda^{1,n},\lambda^{2,n})$ starts from the stationary distribution
$\mu^n_S$, then for any $t\geq 0$, $0\leq s_1\leq \ldots \leq s_k$,
\begin{eqnarray}\label{eqn: stat_lambda_n}
\left(\lambda^{1,n}_{t+s_1},\lambda^{2,n}_{t+s_1}; \ldots ;\lambda^{1,n}_{t+s_k},\lambda^{2,n}_{t+s_k}
\right) \stackrel{d}{=} \left(\lambda^{1,n}_{s_1},\lambda^{2,n}_{s_1}; \ldots
;\lambda^{1,n}_{s_k},\lambda^{2,n}_{s_k} \right).
\end{eqnarray}

Since $\left(\lambda^{1,n}_t,\lambda^{2,n}_t \right)$ converges to $\left(\lambda^1_t,\lambda^2_t
\right)$ pathwise implying the convergence of finite dimensional distribution.
As $n \rr \infty$,
\begin{eqnarray*}
\left(\lambda^{1,n}_{t+s_1},\lambda^{2,n}_{t+s_1}; \ldots ;\lambda^{1,n}_{t+s_k},\lambda^{2,n}_{t+s_k}
\right)  &\Rightarrow & \left(\lambda^1_{t+s_1},\lambda^2_{t+s_1}; \ldots
;\lambda^1_{t+s_k},\lambda^2_{t+s_k} \right), \quad n \rr\infty,\\
\left(\lambda^{1,n}_{s_1},\lambda^{2,n}_{s_1}; \ldots ;\lambda^{1,n}_{s_k},\lambda^{2,n}_{s_k}
\right)  &\Rightarrow & \left(\lambda^1_{s_1},\lambda^2_{s_1}; \ldots ;\lambda^1_{s_k},\lambda^2_{s_k}
\right), \quad n \rr\infty.
\end{eqnarray*}
where the left land side distribution is $\pi^m_S$.

By (\ref{eqn: stat_lambda_n}) and uniqueness of the weak limit, we have the
limiting process
\begin{eqnarray*}
\left(\lambda^{1}_{t+s_1},\lambda^{2}_{t+s_1}; \ldots ;\lambda^{1}_{t+s_k},\lambda^{2}_{t+s_k}
\right) \stackrel{d}{=} \left(\lambda^{1}_{s_1},\lambda^{2}_{s_1}; \ldots
;\lambda^{1}_{s_k},\lambda^{2}_{s_k} \right).
\end{eqnarray*}
i.e. the finite dimensional distribution is independent of $t$. Hence $(\lambda^1,\lambda^2)$
has a stationary distribution $\mu^*_S$.

Since the limiting distribution $\mu^*_A$ exists and is independent with
the initial value, then $\mu^*_S\stackrel{d}{=} \mu^*_A$. As $\mu^*_A$
is unique, the uniqueness of $\mu^*_S$ follows.

\comment{
\begin{eqnarray*}
&& \Prob\left(\lambda^n_{t+s_1} \leq x_1,\ldots, \lambda^n_{t+s_k}\leq x_k
\right) \nonumber\\
&=&\Prob\left(\sum_{i=1}^n \lambda^{(i)}_{t+s_1}  \leq x_1,\ldots, \sum_{i=1}^n
\lambda^{(i)}_{t+s_k} \leq x_k \right)\nonumber\\
&=& \int_{D_1}\cdots \int_{D_k} d\Prob\left(\lambda^{(1)}_{t+s_1}\leq z^1_1,\ldots
\lambda^{(n)}_{t+s_1}\leq z^n_1; \cdots ; \lambda^{(1)}_{t+s_k}\leq z^1_k,\ldots
\lambda^{(n)}_{t+s_k}\leq z^n_k    \right)
\end{eqnarray*}
where $D_j = \{(z_1^j,\ldots,z_n^j) \in \mathbb{R}^n : \sum_{i=1}^n z^j_i
\leq x_j\}$ for $j=1,\ldots, k$.

As $\left(\lambda^{(1)}_t,\ldots,\lambda^{(n)}_t\right)$ is stationary, i.e.
$\Prob\left(\lambda^{(1)}_{t+s_1}\leq z^1_1,\ldots \lambda^{(n)}_{t+s_1}\leq
z^n_1; \cdots ; \lambda^{(1)}_{t+s_k}\leq z^1_k,\ldots \lambda^{(n)}_{t+s_k}\leq
z^n_k \right)$ is independent of $t$. Hence, we have $\Prob\left(\lambda^n_{t+s_1}
\leq x_1,\ldots, \lambda^n_{t+s_k}\leq x_k \right) $ is independent of $t$.
}
\end{proof}

\comment{
\begin{remark}
Denote the stationary distribution of $\left(\lambda^{1,n}_t,\lambda^{2,n}_t\right)$
as $\mu_n$ and $\left(\lambda^{1}_t,\lambda^{2}_t\right)$ as $\mu$, then
we have $\mu_n \Rightarrow \mu$. 
\end{remark}
}

\begin{remark}
Note that the Theorem~7 in Br\'emaud and Massouli\'e\cite{BM1996} has a similar
result  that is a special case of ours and the result can be recovered by
$(C1)$ for the bivariate case.   
\end{remark}

\begin{remark}
From the analysis above, besides the BDCP system $(\lambda^1,\lambda^2)$,
we also provide the distribution
of non-stationary and stationary version of  $(\lambda^{1,n},\lambda^{2,n})$
in terms of Laplace transform. Note that as one can choose the extent of
the contagion effect for the modelling purpose
using the finite system $(\lambda^{1,n},\lambda^{2,n})$. Therefore, it is
an interesting
process itself for applications and further analysis. 
\end{remark}

For any $h>0$, $N_{t_1+h} - N_{t_1}|_{\lambda_{t_1}=\lambda}\stackrel{d}{=}N_{t_2+h}
- N_{t_2}|_{\lambda_{t_2}=\lambda}$. If $\lambda_{t}$ is stationary, then
$\lambda_{t_1}\stackrel{d}{=} \lambda_{t_2}$, and  $N_{t_1+h} - N_{t_1}\stackrel
{d}{=}N_{t_2+h} - N_{t_2}$. Hence we have the following results:
\begin{corollary}\label{cor: Stationarity_BDCP_N}
The BDCP $N$ has stationary increments on $\R_+$.
\end{corollary}

\section{Stationary Moments} \label{sec: StationaryMoments}

From Theorem~\ref{thm: stationarity_BDCP}, $(\lambda^1_t,\lambda^2_t)$ has
a unique stationary distribution $\mu_S^{*}$.
By Proposition 9.2, Chapter 4 in Ethier and Kurtz \cite{EK1986}, for any
$f\in \mathcal{D}(\A)$, 
we have  
\begin{equation}\label{eqn: StatConditionForMoments}
\int_0^\infty \int_0^\infty \A f(\lambda_1,\lambda_2)\mu_S^{*}(\lambda_1,\lambda_2)d\lambda_1
d\lambda_2 = 0.
\end{equation}

In the following, we will use (\ref{eqn: StatConditionForMoments}) to derive
stationary mean, variance and correlation.

\subsection{Stationary Mean}
Take $\A f(\lambda_1,\lambda_2) = \lambda_i$ and denote $m^i_1 = \E[\lambda^i_t]$
as the stationary mean, we have
\begin{eqnarray*}
-(\delta_1-\mu_{G^{1,1}})m^1_1 + \mu_{G^{1,2}}m^2_1 + \rho_1 \mu_{H^1} &=&
0\\
-(\delta_2-\mu_{G^{2,2}})m^2_1 + \mu_{G^{2,1}}m^1_1 + \rho_2 \mu_{H^2} &=&
0.
\end{eqnarray*} 

Solving this linear equation system we obtain the stationary mean as 
\begin{eqnarray*}
m^1_1 &=&\frac{(\delta_2 -\mu_{G^{2,2}})\mu_{H^1}}{(\delta_1 - \mu_{G^{1,1}})(\delta_2
-\mu_{G^{2,2}})-\mu_{G^{1,2}}\mu_{G^{2,1}}}\rho_1
+\frac{\mu_{G^{1,2}}\mu_{H^2}}{(\delta_1 - \mu_{G^{1,1}})(\delta_2 -\mu_{G^{2,2}})-\mu_{G^{1,2}}\mu_{G^{2,1}}}\rho_2
\\
m^2_1 &=&\frac{(\delta_1 -\mu_{G^{1,1}})\mu_{H^2}}{(\delta_1 - \mu_{G^{1,1}})(\delta_2
-\mu_{G^{2,2}})-\mu_{G^{1,2}}\mu_{G^{2,1}}}\rho_2
+\frac{\mu_{G^{2,1}}\mu_{H^1}}{(\delta_1 - \mu_{G^{1,1}})(\delta_2 -\mu_{G^{2,2}})-\mu_{G^{1,2}}\mu_{G^{2,1}}}\rho_1.\label{ean:
statMean}
\end{eqnarray*}
Denote $\Delta_i = \delta_i - \mu_{G^{i,i}}$ for $i=1,2$, and $\Delta :=
\Delta_1\Delta_2 - \mu_{G^{1,2}}\mu_{G^{2,1}}$, then we can rewrite the first
moments as
\begin{eqnarray*}
m_1^1 &=& \frac{\Delta_2 \mu_{H^1}}{\Delta} \rho_1 + \frac{\mu_{G^{1,2}}\mu_{H^2}}{\Delta}
\rho_2 =: \mu_{1,1}\rho_1 + \mu_{1,2}\rho_2\\
m_1^2 &=& \frac{\mu_{G^{2,1}}\mu_{H^1}}{\Delta}\rho_1 + \frac{\Delta_1 \mu_{H^2}}{\Delta}\rho_2
=: \mu_{2,1}\rho_1 + \mu_{2,2}\rho_2.
\end{eqnarray*}
\begin{remark}
If there is no cross-exciting term, i.e. $\mu_{G^{i,j}} = 0$ for $i\neq j$,
then the result recovers univariate DCP in Dassios and Zhao\cite{DZ2011}.

\end{remark}

\subsection{Stationary Variance}
We consider the stationary moments of $(\lambda_t^1)^2$, $(\lambda_t^2)^2$
and $\lambda_t^1 \lambda_t^2$.

Take $f(t,\lambda^1_t,\lambda^2_t) = (\lambda^1_t)^2$ and $f(t,\lambda^1_t,\lambda^2_t)
= (\lambda^2_t)^2$, we have
\begin{eqnarray*}
\A (\lambda_1^2) &=& -2\delta_1 \lambda_1^2 + \rho_1\left[ \int_0^\infty
(\lambda_1+y_1)^2H^1(dy_1) -\lambda_1^2 \right] \\
&& + \lambda_1 \left[ \int_0^\infty (\lambda_1 + z_1)^2 G^{1,1}(dz_1) - \lambda_1^2
\right] + \lambda_2 \left[ \int_0^\infty (\lambda_1 + z_1)^2 G^{1,2}(dz_1)
- \lambda_1^2 \right]\\
&=& -2\delta_1 \lambda_1^2 + \rho_1 (2\lambda_1 \mu_{H^1}+ \mu_{2H^1}) +
\lambda_1 (2\lambda_1 \mu_{G^{1,1}} + \mu_{2G^{1,1}}) + \lambda_2 (2\lambda_1
\mu_{G^{1,2}}+ \mu_{2G^{1,2}})\\
&=&  -2(\delta_1- \mu_{G^{1,1}}) \lambda_1^2 + 2 \mu_{G^{1,2}} \lambda_1
\lambda_2 + (2\rho_1 \mu_{H^1}+\mu_{2G^{1,1}})\lambda_1 + \mu_{2G^{1,2}}
\lambda_2 +  \mu_{2H^1} \rho_1.
\end{eqnarray*}

Similarly, we have
\begin{eqnarray*}
\A (\lambda_2^2) &=&  -2(\delta_2 - \mu_{G^{2,2}}) \lambda_2^2 + 2 \mu_{G^{2,1}}
\lambda_1 \lambda_2 + (2\rho_2 \mu_{H^2}+\mu_{2G^{2,2}})\lambda_2 + \mu_{2G^{2,1}}
\lambda_1 +  \mu_{2H^2} \rho_2\\
\A \lambda_1 \lambda_2 &=& - (\delta_1+\delta_2) \lambda_1\lambda_2 + \rho_1
\lambda_2 \mu_{H^1} + \rho_2\lambda_1 \mu_{H^2} + \lambda_1 \left( \lambda_1
\mu_{G^{2,1}} + \lambda_2 \mu_{G^{1,1}} + \mu_{G^{1,1}}\mu_{G^{2,1}} \right)\\
&& + \lambda_2 (\lambda_1 \mu_{G^{2,2}} + \lambda_2 \mu_{G^{1,2}} + \mu_{G^{1,2}}
\mu_{G^{2,2}})\\
&=&  \mu_{G^{2,1}} \lambda_1^2 + \mu_{G^{1,2}} \lambda_2^2 + \left( - (\delta_1-
\mu_{G^{1,1}})- (\delta_2 - \mu_{G^{2,2}}) \right) \lambda_1\lambda_2 \\
&& + (\rho_2 \mu_{H^2} + \mu_{G^{1,1}}\mu_{G^{2,1}})\lambda_1 + (\rho_1\mu_{H^1}+\mu_{G^{1,2}}
\mu_{G^{2,2}})\lambda_2. 
\end{eqnarray*}

We can rewrite
\begin{eqnarray*}
\A (\lambda_1^2) 
&=:& A_{1,1}\lambda_1^2 + A_{1,2} \lambda_1\lambda_2 +  A_1 \lambda_1 + A_{2}\lambda_2
+ A_0\\
\A (\lambda_2^2)
&=:&  B_{2,2}\lambda_2^2 + B_{1,2} \lambda_1\lambda_2 +  B_2 \lambda_2  +
B_{1}\lambda_1 + B_0\\
\A \lambda_1 \lambda_2 
&=:& C_{1,1} \lambda_1^2 + C_{2,2}\lambda_2^2 + C_{1,2} \lambda_1\lambda_2
+ C_1 \lambda_1 + C_2 \lambda_2.
\end{eqnarray*}
with all coefficients in Table~\ref{tab: momentCoeffABC}. Note that $A_{i,j},B_{i,j},C_{i,j}$
do not contain $\rho$ and $A_1, B_2, C_1,C_2$ are linear with $\rho$. 

\begin{table}
\begin{center}
\begin{tabular}{|r|r|r|r|r|r|l|}
  \hline
$X$  & $X_{1,1}$ & $X_{2,2}$ & $X_{1,2}$ & $X_1$ & $X_2$ & $X_0$ \\
 \hline
$A$ & $-2\Delta_1$ & $0$ & $2\mu_{G^{1,2}}$ & $2\mu_{H^1}\rho_1 + \mu_{2G^{1,1}}$
& $\mu_{2G^{1,2}}$ & $\mu_{2H^1}\rho_1$ \\
  \hline
$B$ & $0$ & $-2\Delta_2$ & $2\mu_{G^{2,1}}$ & $\mu_{2G^{2,1}}$ & $2\mu_{H^2}\rho_2
+ \mu_{2G^{2,2}}$ & $\mu_{2H^2} \rho_2$ \\
\hline
$C$ & $\mu_{G^{2,1}}$ & $\mu_{G^{1,2}}$  & $-\Delta_1 - \Delta_2$ & $\mu_{H^2}\rho_2
+ \mu_{G^{1,1}}\mu_{G^{2,1}}$ & $\mu_{H^1}\rho_1 + \mu_{G^{1,2}}\mu_{G^{2,2}}$
& $0$\\
\hline  
\end{tabular}
\caption{coefficient table for $A,B,C$, $\Delta_1 = \delta_1 - \mu_{G^{1,1}}
>0$ and $\Delta_2 = \delta_2 - \mu_{G^{2,2}}>0$.}
\end{center}\label{tab: momentCoeffABC}
\end{table}

We denote 
\begin{eqnarray*}
m_2^{i} &=:&\E[(\lambda^i_t)^2] \\
 m_2^{i,j} &=:& \E[\lambda^i_t \lambda^j_t],
\end{eqnarray*}
then by (\ref{eqn: StatConditionForMoments}) we obtain the linear equation
system:
\begin{eqnarray*}
A_{1,1}m^1_2 + A_{1,2}m ^{1,2}_2 + \left( A_1 m_1^1 + A_2 m_1^2 + A_0 \right)
&=& 0\\
B_{2,2} m^2_2 + B_{1,2}m^{1,2}_2 + \left( B_2 m_1^2+ B_1 m_1^1 + B_0 \right)&=&
0\\
C_{1,1} m^1_2 + C_{2,2} m_2^2 + C_{1,2} m^{1,2}_2 + \left (C_1 m^1_1 + C_2
m^2_1 \right) &=& 0.
\end{eqnarray*}

Note that the cross term is
\begin{equation}\label{eqn: crossMoment12}
m_2^{1,2} = -\frac{C_1 m^1_1 + C_2 m^2_1 + C_{1,1} m^1_2 + C_{2,2} m_2^2
}{C_{1,2}}.
\end{equation}
Based on the first moments $m_1^1$, $m_1^2$ obtained in the last section,
we obtain the second moments $m_2^1$ and $m_2^2$ by solving the linear equation
system. 
\begin{eqnarray*}
\left(A_{1,1} - A_{1,2}\frac{C_{1,1}}{C_{1,2}} \right) m_2^1 -A_{1,2}\frac{C_{2,2}}{C_{1,2}}
m_2^2 + \left(\td{A}_0 - \frac{A_{1,2}}{C_{1,2}}\td{C}_0 \right) &=& 0\\
\left(B_{2,2}-B_{1,2}\frac{C_{2,2}}{C_{1,2}} \right)m_2^2 - B_{1,2}\frac{C_{1,1}}{C_{1,2}}
m_2^1 + \left(\td{B}_0 - \frac{B_{1,2}}{C_{1,2}}\td{C}_0 \right) &=& 0,
\end{eqnarray*}
where 
\begin{eqnarray*}
\td{A}_0 &=& A_1 m_1^1 + A_2 m_1^2 + A_0\\
\td{B}_0 &=& B_2 m_1^2+ B_1 m_1^1 + B_0\\
\td{C}_0 &=& C_1 m^1_1 + C_2 m^2_1.
\end{eqnarray*}

Denote $\gamma_1 := \frac{C_{1,1}}{C_{1,2}}$ and $\gamma_2 := \frac{C_{2,2}}{C_{1,2}}$,
then
\begin{eqnarray*}
\left[A_{1,1}-A_{1,2}\gamma_1 - \frac{A_{1,2}B_{1,2}\gamma_1\gamma_2}{B_{2,2}-B_{1,2}\gamma_2}\right]
m_2^1 
= -\frac{(\td{B}_0 - \frac{B_{1,2}}{C_{1,2}}\td{C}_0)A_{1,2}\gamma_2}{B_{2,2}-B_{1,2}\gamma_2}-\left(
\td{A}_0 -\frac{A_{1,2}}{C_{1,2}}\td{C}_0\right).
\end{eqnarray*}

\begin{eqnarray*}
m_2^1 &=& \frac{ -\frac{(\td{B}_0 - \frac{B_{1,2}}{C_{1,2}}\td{C}_0)A_{1,2}\gamma_2}{B_{2,2}-B_{1,2}\gamma_2}-\left(
\td{A}_0 - \frac{A_{1,2}}{C_{1,2}}\td{C}_0\right)}{A_{1,1}-A_{1,2}\gamma_1
- \frac{A_{1,2}B_{1,2}\gamma_1\gamma_2}{B_{2,2}-B_{1,2}\gamma_2}}= \frac{(B_{1,2}\gamma_2
-B_{2,2})\td{A}_0 - A_{1,2}\gamma_2 \td{B}_0 +
\frac{A_{1,2}}{C_{1,2}}B_{2,2}\td{C}_0}{4(\Delta_1\Delta_2 - \mu_{G^{1,2}}\mu_{G^{2,1}})}.
\end{eqnarray*}

Similarly, we have 
\begin{eqnarray*}
m_2^2 &=& \frac{-B_{1,2}\gamma_1\td{A}_0 - (A_{1,1}-A_{1,2}\gamma_1) \td{B}_0
+ \frac{B_{1,2}}{C_{1,2}}A_{1,1}\td{C}_0}{4(\Delta_1\Delta_2 - \mu_{G^{1,2}}\mu_{G^{2,1}})}.
\end{eqnarray*}

We obtain 
\begin{eqnarray*}
m_2^1 &=& (m_1^1)^2 + \gamma_{1,1} \rho_1 + \gamma_{1,2} \rho_2\\
m_2^2 &=& (m_1^2)^2 + \gamma_{2,1} \rho_1 + \gamma_{2,2} \rho_2,  
\end{eqnarray*}
where 
\begin{eqnarray*}
\gamma_{1,1} &=& \frac{1}{2\Delta}\left( \frac{-2\mu_{G^{2,1}}\mu_{G^{1,2}}}{\Delta_1
+ \Delta_2} + \Delta_2 \right)\left( \mu_{2G^{1,1}}\mu_{1,1} + \mu_{2G^{1,2}}\mu_{2,1}+\mu_{2H^1}
\right) \nonumber \\
&& + \frac{1}{2\Delta} \frac{(\mu_{G^{1,2}})^2}{\Delta_1 + \Delta_2} \left(
\mu_{2G^{2,2}}\mu_{2,1} + \mu_{2G^{2,1}}\mu_{1,1} \right)
 + \frac{1}{\Delta} \frac{\mu_{G^{1,2}}\Delta_2}{\Delta_1 + \Delta_2}\left(
\mu_{G^{1,1}}\mu_{G^{2,1}}\mu_{1,1} + \mu_{G^{1,2}}\mu_{G^{2,2}}\mu_{2,1}\right)
\nonumber\\
\gamma_{1,2} &=& \frac{1}{2\Delta}\left( \frac{-2\mu_{G^{2,1}}\mu_{G^{1,2}}}{\Delta_1
+ \Delta_2} + \Delta_2 \right)\left( \mu_{2G^{1,1}}\mu_{1,2} + \mu_{2G^{1,2}}\mu_{2,2}\right)
\nonumber\\
&& + \frac{1}{2\Delta} \frac{(\mu_{G^{1,2}})^2}{\Delta_1 + \Delta_2} \left(
\mu_{2G^{2,2}}\mu_{2,2} + \mu_{2G^{2,1}}\mu_{1,2} + \mu_{2H^2}   \right)
 + \frac{1}{\Delta} \frac{\mu_{G^{1,2}}\Delta_2}{\Delta_1 + \Delta_2}\left(
\mu_{G^{1,1}}\mu_{G^{2,1}}\mu_{1,2} + \mu_{G^{1,2}}\mu_{G^{2,2}}\mu_{2,2}\right).
\label{eqn: gamma_11} 
\end{eqnarray*}

Similarly, 
\begin{eqnarray*}
\gamma_{2,1} &=& \frac{1}{2\Delta}\left( \frac{-2\mu_{G^{2,1}}\mu_{G^{1,2}}}{\Delta_1
+ \Delta_2} + \Delta_1 \right)\left(\mu_{2G^{2,2}}\mu_{2,1} + \mu_{2G^{2,1}}\mu_{11}
\right) \nonumber \\
&& + \frac{1}{2\Delta} \frac{(\mu_{G^{2,1}})^2}{\Delta_1 + \Delta_2} \left(
\mu_{2G^{1,1}}\mu_{1,1} + \mu_{2G^{1,2}}\mu_{2,1} + \mu_{2H^1} \right)
 + \frac{1}{\Delta} \frac{\mu_{G^{2,1}}\Delta_1}{\Delta_1 + \Delta_2}\left(
\mu_{G^{1,1}}\mu_{G^{2,1}}\mu_{1,1} + \mu_{G^{1,2}}\mu_{G^{2,2}}\mu_{2,1}\right)
\nonumber\\
\gamma_{2,2} &=& \frac{1}{2\Delta}\left( \frac{-2\mu_{G^{2,1}}\mu_{G^{1,2}}}{\Delta_1
+ \Delta_2} + \Delta_1 \right)\left( \mu_{2G^{2,2}}\mu_{2,2} + \mu_{2G^{2,1}}\mu_{1,2}
+ \mu_{2H^2}\right) \nonumber\\
&& + \frac{1}{2\Delta} \frac{(\mu_{G^{2,1}})^2}{\Delta_1 + \Delta_2} \left(
\mu_{2G^{1,1}}\mu_{1,2} + \mu_{2G^{1,2}}\mu_{2,2} \right) 
 + \frac{1}{\Delta} \frac{\mu_{G^{2,1}}\Delta_1}{\Delta_1 + \Delta_2}\left(
\mu_{G^{1,1}}\mu_{G^{2,1}}\mu_{1,2} + \mu_{G^{1,2}}\mu_{G^{2,2}}\mu_{2,2}\right).
\label{eqn: gamma_22} 
\end{eqnarray*}

Hence,we conclude the stationary mean and variance are
\begin{eqnarray*}
m_1 &:=& \E[\lambda_t^1] = \mu_{1,1} \rho_1 + \mu_{1,2} \rho_2\\
m_2 &:=& \E[\lambda_t^2] =  \mu_{2,1} \rho_1 + \mu_{2,2} \rho_2
\end{eqnarray*}

and
\begin{eqnarray*}
v_1 &:=& \Var(\lambda_t^1) = \gamma_{1,1} \rho_1 + \gamma_{1,2} \rho_2\\
v_2 &:=& \Var(\lambda_t^2) =  \gamma_{2,1} \rho_1 + \gamma_{2,2} \rho_2
\end{eqnarray*}
with coefficients above.

\begin{remark}
From above, we observe that the stationary mean and variance are both linear
functions of $\rho_1$ and $\rho_2$.
\end{remark}

\subsection{Stationary Correlation} 

\begin{eqnarray*}
\rho_{1,2} = \frac{E[\lambda^1_t \lambda^2_t ] - \E[\lambda^1_t]\E[\lambda^2_t]}{\sqrt{\Var(\lambda^1_t)}\sqrt{\Var(\lambda^2_t)}}
=  \frac{m_2^{1,2} - m_1 m_2 }{\sqrt{v_1} \sqrt{v_2}},
\end{eqnarray*}
where $m^{1,2}_2$ is from (\ref{eqn: crossMoment12}).

Note that the stationary correlation is larger than processes with only self-exciting
jumps as cross-exciting jumps have positive mean $\mu_{G^{1,2}}$ and $\mu_{G^{2,1}}$.

\section{Conclusion}  \label{sec: Conclusion}

By using the the Markov theory on the branching system approximation, we
found the condition under which there exists a unique stationary distribution
of the BDCP intensity and the resulting BDCP has stationary increments. All
moments of the stationary intensity can be computed using the Markov property.
Moreover, we also obtained  the limiting and stationary distributions of
the approximating sequence of the intensity in terms of the Laplace transform
which is also useful in practice.

\bibliographystyle{apt}

\appendix

\section{Proof of Lemma}

\subsection{Proof of Lemma~\ref{lem: lfunction_asympt_finiteSystem}}\label{sec:
proof_lem_lfunction_asympt_finiteSystem}
\begin{proof}
First, we have
$$
\lim_{t\rr \infty}l_1(t) = \lim_{t \rr \infty} v_{2n}e^{-\delta_2 t} =
0,\quad \lim_{t\rr \infty}l_2(t) = \lim_{t \rr \infty} v_{2n-1} e^{-\delta_1
t}
= 0.$$

Then, assume $\lim_{t\rr \infty}l_{2k-1}(t) = 0$ and $\lim_{t\rr \infty}
l_{2k}(t) = 0$, then 
\begin{eqnarray*}
\lim_{t\rr \infty} l_{2k+1}(t) &=& \lim_{t\rr \infty}  e^{-\delta_2 t}\int_0^t
e^{\delta_2 s}\left[ 1-\hat{g}_{1,2}(l_{2k}(s))+1-\hat{g}_{2,2}(l_{2k-1}(s))\right]
ds \\
&\stackrel{L'Hospital}{=}& \lim_{t\rr \infty} \frac{1}{\delta_2} \left(1-\hat{g}_{1,2}(l_{2k}(t))+1-\hat{g}_{2,2}(l_{2k-1}(t))\right)\\
&=& 0. 
\end{eqnarray*}
Similarly, we have $\lim_{t\rr \infty} l_{2k+2}(t) = 0 $. 

Hence, by induction, we conclude that for any $i = 1,\dots,m$, $\lim_{t\rr
\infty}l_i(t) = 0
$.

\end{proof}

\subsection{Proof of Lemma~\ref{thm: statCondTruncSys}}\label{sec: ProofThm_statCondTruncSys}
\begin{proof}
By Proposition 9.2, Chapter 4 in Ethier and Kurtz \cite{EK1986}, the stationary
distribution $\mu$ exists if and only if for $f\in\mathcal{D}(\A)$, $\int
\A f d\mu = 0$. Hence we show the stationary distribution $\pi$ satisfies
\begin{eqnarray}\label{eqn: stationaryIntegralEqn}
\int \A_m f(\lambda_1,\lambda_2,\ldots,\lambda_m) \pi(\lambda_1,\lambda_2,\ldots,\lambda_m)
d\lambda_1\cdots d\lambda_n = 0.
\end{eqnarray}

We now derive the equivalent Laplace transform equation.

Part (I): drift part
\begin{eqnarray*}
&&\int_{\R_+^m} -\delta_k \lambda_k \frac{\partial }{\partial\lambda_k}f(\lambda_1,\ldots,\lambda_m)\pi(\lambda_1,\ldots,\lambda_m)d\lambda_1\cdots
d\lambda_m \\
&=&  -\delta_k \int_{\R_+^{m}}    \frac{\partial }{\partial\lambda_k}f(\lambda_1,\ldots,\lambda_m)\int_0^{\lambda_k}
\frac{\partial}{\partial \lambda_k}\left( x \pi(\lambda_1,\ldots,x,\ldots,\lambda_m)
\right)d x d\lambda_1\cdots d\lambda_m \\
&=& -\delta_k \int_{\R_+^{m-1}} \int_{\lambda_k = 0}^\infty \int_{x=0}^{\lambda_k}
  \frac{\partial }{\partial\lambda_k}f(\lambda_1,\ldots,\lambda_m) \frac{\partial}{\partial
\lambda_k}\left( x \pi(\lambda_1,\ldots,x,\ldots,\lambda_m) \right)d x d\lambda_1\cdots
d\lambda_m \\
&=& -\delta_k \int_{\R_+^{m-1}} \int_{x = 0}^\infty \int_{\lambda_k = x}^{\infty}
  \frac{\partial }{\partial\lambda_k}f(\lambda_1,\ldots,\lambda_m) \frac{\partial}{\partial
\lambda_k}\left( x \pi(\lambda_1,\ldots,x,\ldots,\lambda_m) \right)d x d\lambda_1\cdots
d\lambda_m \\
&=& \int_{\R_+^m} f(\lambda_1,\ldots,\lambda_m)  \delta_k \frac{\partial}{\partial
\lambda_k}\left( \lambda_k \pi(\lambda_1,\ldots,\lambda_m) \right)d\lambda_1\cdots
d\lambda_m,
\end{eqnarray*}
where we have used the fact that $f(\lambda_1,\ldots,\lambda_m)|_{\lambda_k
= \infty} = 0$.

Take the Laplace transform at $(v_1,\ldots,v_m)$ is
\begin{eqnarray*}
\mathcal{L}_m\left[  \delta_k \frac{\partial}{\partial \lambda_k}\left( \lambda_k
\pi(\lambda_1,\ldots,\lambda_m) \right) \right] =  \delta_k v_k \Lp_m [\lambda_k
\pi(\lambda_1,\ldots,\lambda_m) ] = -  \delta_k v_k \frac{\partial }{\partial
v_k}\hat{\pi}(v_1,\ldots,v_m).
\end{eqnarray*}

Part (II): shot-noise part
\begin{eqnarray*}
&& \int_{\R^m_+} \rho_1 \int_{y=0}^\infty f(\lambda_1+y,\lambda_2,\ldots,\lambda_m)dH_1(y)\pi(\lambda_1,\ldots,\lambda_m)
d\lambda_1\cdots d\lambda_m \\
&=&  \rho_1 \int_{\R^{m-1}_+} \int_{x=0}^\infty f(x,\lambda_2,\ldots,\lambda_m)
\int_{y=0}^{x} \pi(x-y,\lambda_2,\ldots,\lambda_m)dH_1(y) dx d\lambda_2\cdots
d\lambda_m \\
&=&  \rho_1 \int_{\R^{m}_+}  f(\lambda_1,\lambda_2,\ldots,\lambda_m) \int_{y=0}^{\lambda_1}
\pi(\lambda_1-y,\lambda_2,\ldots,\lambda_m)dH_1(y) d\lambda_1\cdots d\lambda_m,
\end{eqnarray*}

then
\begin{eqnarray*}
&& \int_{\R^m_+} \rho_1 \left[\int_0^\infty f(\lambda_1+y,\lambda_2,\ldots,\lambda_m)dH_1(y)
- f(\lambda_1,\ldots,\lambda_m) \right] \pi(\lambda_1,\ldots,\lambda_m) d\lambda_1\cdots
d\lambda_m \\
&=& \rho_1 \int_{\R^{m}_+}  f(\lambda_1,\lambda_2,\ldots,\lambda_m)\left[
\int_{y=0}^{\lambda_1} \pi(\lambda_1-y,\lambda_2,\ldots,\lambda_m)dH_1(y)
- \pi(\lambda_1,\ldots,\lambda_m)\right] d\lambda_1\cdots d\lambda_m.
\end{eqnarray*}

We have the Laplace transform as 
\begin{eqnarray*}
\mathcal{L}_m\left[  \int_{y=0}^{\lambda_1} \pi(\lambda_1-y,\lambda_2,\ldots,\lambda_m)dH_1(y)
 \right] = \hat{\pi}(v_1,\ldots,v_m) \hat{h}(v_1).  
\end{eqnarray*}

Similarly, we have 
\begin{eqnarray*}
\mathcal{L}_m\left[  \int_{y=0}^{\lambda_2} \pi(\lambda_1,\lambda_2-y,\ldots,\lambda_m)dH_1(y)
 \right] = \hat{\pi}(v_1,\ldots,v_m) \hat{h}(v_2).  
\end{eqnarray*}

Part(III): exciting part: 

For $k\geq 2$, the jump excited by $\lambda_{2k-1}$ is
\begin{eqnarray*}
&& \int_{\R^m_+} \lambda_{2k-1}\left[  \left(\int_{z=0}^\infty f(\cdot,\lambda_{2k+1}
+ z,\cdot)dG_{1,1}(z) - f(\ldots) \right) + \left(\int_{z=0}^\infty f(\cdot,\lambda_{2k+2}
+ z,\cdot)dG_{2,1}(z) - f(\ldots,) \right) \right]\\
&& \cdot\pi(\lambda_1,\ldots,\lambda_m) d\lambda_1\cdots d\lambda_m \\
&=&  \int_{\R^{m}_+}  f(\lambda_1,\ldots,\lambda_m) \lambda_{2k-1} \left[\int_{z=0}^{\lambda_{2k+1}}
\pi(\lambda_1,\ldots,\lambda_{2k+1}-z,\ldots,\lambda_m)dG_{1,1}(z) - \pi(\lambda_1,\ldots,\lambda_m)\right]
d\lambda_1\cdots d\lambda_n\\
&& + \int_{\R^{m}_+}  f(\lambda_1,\ldots,\lambda_m) \lambda_{2k-1} \left[\int_{z=0}^{\lambda_{2k+2}}
\pi(\lambda_1,\ldots,\lambda_{2k+2}-z,\ldots,\lambda_m)dG_{2,1}(z) - \pi(\lambda_1,\ldots,\lambda_m)\right]
d\lambda_1\cdots d\lambda_n\\
\end{eqnarray*}

We have first 
\begin{eqnarray*}
&& \mathcal{L}_m\left[ \lambda_{2k-1} \int_{z=0}^{\lambda_{2k+1}} \pi(\lambda_1,\ldots,\lambda_{2k+1}-z,\ldots,\lambda_m)dG_{1,1}(z)
-\lambda_{2k-1} \pi(\lambda_1,\ldots,\lambda_m)  \right]\\
&=& \mathcal{L}_m \left[\lambda_{2k-1}\pi(\lambda_1,\ldots,\ldots,\lambda_m)
\right]  \hat{g}_{1,1}(v_{2k+1}) - \mathcal{L}_m \left[\lambda_{2k-1}\pi(\lambda_1,\ldots,\lambda_m)
\right] \\
&=&\frac{\partial}{\partial v_{2k-1}} \hat{\pi}(v_1,\ldots,v_m)(1- \hat{g}_{1,1}(v_{2k+1})).
\end{eqnarray*}

Similarly,
\begin{eqnarray*}
&& \mathcal{L}_m\left[ \lambda_{2k-1} \int_{z=0}^{\lambda_{2k+2}} \pi(\lambda_1,\ldots,\lambda_{2k+2}-z,\ldots,\lambda_m)dG_{2,1}(z)
-\lambda_{2k-1} \pi(\lambda_1,\ldots,\lambda_m)  \right]\\
&=&\frac{\partial}{\partial v_{2k-1}} \hat{\pi}(v_1,\ldots,v_m)(1- \hat{g}_{2,1}(v_{2k+2})).
\end{eqnarray*}

The same for jumps excited by $\lambda_{2k}$ corresponding to 
\begin{eqnarray*}
&&\int_{\R^m_+} \lambda_{2k}\left[  \left(\int_{z=0}^\infty f(\cdot,\lambda_{2k+1}
+ z,\cdot)dG_{1,2}(z) - f(\ldots) \right) + \left(\int_{z=0}^\infty f(\cdot,\lambda_{2k+2}
+ z,\cdot)dG_{2,2}(z) - f(\ldots,) \right) \right]\nonumber\\
&&\cdot\pi(\lambda_1,\ldots,\lambda_m) d\lambda_1\cdots d\lambda_m.
\end{eqnarray*}

Since stationary distribution $\pi$ satisfies (\ref{eqn: stationaryIntegralEqn}),
we have from part (I), (II), (III) that
\begin{eqnarray*}
0 &=& \sum_{k=1}^m\delta_k \frac{\partial}{\partial \lambda_k}\left( \lambda_k
\pi(\lambda_1,\ldots,\lambda_m) \right)\\
&& + \rho_1 \left[ \int_{y=0}^{\lambda_1} \pi(\lambda_1-y,\lambda_2,\ldots,\lambda_m)dH_1(y)
- \pi(\lambda_1,\ldots,\lambda_m)\right]\\
&& + \rho_2 \left[ \int_{y=0}^{\lambda_2} \pi(\lambda_1,\lambda_2-y,\ldots,\lambda_m)dH_2(y)
- \pi(\lambda_1,\ldots,\lambda_m)\right]\\
&& +  \sum_{k=1}^{n-1}\lambda_{2k-1} \left[\int_{z=0}^{\lambda_{2k+1}} \pi(\lambda_1,\ldots,\lambda_{2k+1}-z,\ldots,\lambda_m)dG_{1,1}(z)
- \pi(\lambda_1,\ldots,\lambda_m)\right]\\
&& +  \sum_{k=1}^{n-1}\lambda_{2k-1} \left[\int_{z=0}^{\lambda_{2k+2}} \pi(\lambda_1,\ldots,\lambda_{2k+2}-z,\ldots,\lambda_m)dG_{2,1}(z)
- \pi(\lambda_1,\ldots,\lambda_m)\right]\\
&& + \sum_{k=1}^{n-1}\lambda_{2k} \left[\int_{z=0}^{\lambda_{2k+1}} \pi(\lambda_1,\ldots,\lambda_{2k+1}-z,\ldots,\lambda_m)dG_{1,2}(z)
- \pi(\lambda_1,\ldots,\lambda_m)\right]\\
&& +  \sum_{k=1}^{n-1}\lambda_{2k} \left[\int_{z=0}^{\lambda_{2k+2}} \pi(\lambda_1,\ldots,\lambda_{2k+2}-z,\ldots,\lambda_m)dG_{2,2}(z)
- \pi(\lambda_1,\ldots,\lambda_m)\right].
\end{eqnarray*}

In terms of Laplace transform, we have for any $(v_1,\ldots, v_m) \in \R^m_+$
\begin{eqnarray*}
0 &=& -\sum_{k=1}^{2n} \delta_k v_k \frac{\partial \hat{\pi}^m_S}{\partial
v_k}+ \rho_1(\hat{h}(v_1)-1) +\rho_2(\hat{h}(v_2)-1) \nonumber \\
&& + \sum_{k=1}^{n-1} \frac{\partial \hat{\pi}^m_S}{\partial v_{2k-1}} \left[(1-\hat{g}_{1,1}(v_{2k+1}))
+ (1-\hat{g}_{2,1}(v_{2k+2})) \right] 
+ \sum_{k=1}^{n-1} \frac{\partial \hat{\pi}^m_S}{\partial v_{2k}} \left[(1-\hat{g}_{1,2}(v_{2k+1}))
+ (1-\hat{g}_{2,2}(v_{2k+2})) \right]. \nonumber \\
\end{eqnarray*}

Reorder the terms, we have (\ref{eqn: StatEq_v1}).
\end{proof}

\subsection{Proof of (\ref{eqn: distance_function})}\label{sec:
Proof_dRecursiveDominated}\label{sec: proof_distanceFunction}
\begin{proof}

For $j=1,2$,
\begin{eqnarray*}
&&\hat{g}_{1,j}(l_{2k-2}(t))- \hat{g}_{1,j}(l_{2k}(t)) = \int_{l_{2k}(t)}^{l_{2k-2}(t)}
d\hat{g}_{1,j}(u) = \int_{l_{2k-2}(t)}^{l_{2k}(t)}\left(-\hat{g}'_{1,j}(u)\right)du
\leq \mu_{G^{1,j}} \left( l_{2k}(t) - l_{2k-2}(t)\right)\\
&&\hat{g}_{2,j}(l_{2k-3}(t))- \hat{g}_{2,j}(l_{2k-1}(t)) = \int_{l_{2k-1}(t)}^{l_{2k-3}(t)}
d\hat{g}_{2,j}(u) = \int_{l_{2k-3}(t)}^{l_{2k-1}(t)}\left(-\hat{g}'_{2,j}(u)\right)du
\leq \mu_{G^{2,j}} \left( l_{2k-1}(t) - l_{2k-3}(t)\right)\\
\end{eqnarray*}
Then,
\begin{eqnarray*}
d^{(1)}_{k+1}(t) &=&  e^{-\delta_2 t} \int_0^t e^{\delta_2 s} \left[ \left(1-
\hat{g}_{1,2}(l_{2k}(s)) \right) - \left(1- \hat{g}_{1,2}(l_{2k-2}(s)) \right)
\right] ds \\ 
&& +  e^{-\delta_2 t} \int_0^t e^{\delta_2 s} \left[ \left(1- \hat{g}_{2,2}(l_{2k-1}(s))
\right) - \left(1- \hat{g}_{2,2}(l_{2k-3}(s)) \right) \right] ds\\
&\leq & e^{-\delta_2 t} \int_0^t e^{\delta_2 s} \left[ \mu_{G^{2,2}} d^{(1)}_k(s)+
\mu_{G^{1,2}}d^{(2)}_k(s) \right] ds \\
d^{(2)}_{k+1}(t) &=&  e^{-\delta_1 t} \int_0^t e^{\delta_1 s} \left[ \left(1-
\hat{g}_{1,1}(l_{2k}(s)) \right) - \left(1- \hat{g}_{1,1}(l_{2k-2}(s)) \right)
\right] ds \\ 
&& +  e^{-\delta_1 t} \int_0^t e^{\delta_1 s} \left[ \left(1- \hat{g}_{2,1}(l_{2k-1}(s))
\right) - \left(1- \hat{g}_{2,1}(l_{2k-3}(s)) \right) \right] ds \\
&\leq & e^{-\delta_1 t} \int_0^t e^{\delta_1 s} \left[ \mu_{G^{2,1}} d^{(1)}_k(s)
+ \mu_{G^{1,1}}d^{(2)}_k(s) \right] ds 
\end{eqnarray*}
\end{proof}

\end{document}